\documentclass[onecolumn,notitlepage,longbibliography]{revtex4-1}
\usepackage{amsthm,amsmath,amssymb}
\usepackage{latexsym,amssymb,verbatim,enumerate,graphicx}
\usepackage{hyperref}
\usepackage{color}

\newcommand{\calA}{\mathcal{A}}
\newcommand{\calB}{\mathcal{B}}
\newcommand{\calC}{\mathcal{C}}
\newcommand{\calD}{\mathcal{D}}
\newcommand{\calE}{\mathcal{E}}
\newcommand{\calF}{\mathcal{F}}

\newcommand{\calH}{\mathcal{H}}
\newcommand{\calI}{\mathcal{I}}
\newcommand{\calJ}{\mathcal{J}}
\newcommand{\calK}{\mathcal{K}}
\newcommand{\calL}{\mathcal{L}}
\newcommand{\calM}{\mathcal{M}}
\newcommand{\calN}{\mathcal{N}}
\newcommand{\calO}{\mathcal{O}}

\newcommand{\calR}{\mathcal{R}}
\newcommand{\calS}{\mathcal{S}}

\newcommand{\calZ}{\mathcal{Z}}
\theoremstyle{plain} 

\newtheorem{thm}{Theorem}[section]
\newtheorem{cor}[thm]{Corollary}
\newtheorem{lem}[thm]{Lemma}
\newtheorem{prop}[thm]{Proposition}
\theoremstyle{definition}
\newtheorem{defn}[thm]{Definition}
\newtheorem{example}[thm]{Example}
\newtheorem{conj}[thm]{Conjecture}
\newtheorem{remark}[thm]{Remark}

\newcommand{\bthm}{\begin{thm}}
\newcommand{\ethm}{\end{thm}}
\newcommand{\blm}{\begin{lem}}
\newcommand{\elm}{\end{lem}}
\newcommand{\bcor}{\begin{cor}}
\newcommand{\ecor}{\end{cor}}
\newcommand{\bex}{\begin{example}}
\newcommand{\eex}{\end{example}}
\newcommand{\bdf}{\begin{defn}}
\newcommand{\edf}{\end{defn}}
\newcommand{\Tr}{\operatorname{Tr}}
\newcommand{\Ker}{\operatorname{Ker}}
\newcommand{\ot}{\otimes}
\def\Span{{\rm span}}
\def\supp{{\rm supp}}

\def\>{\rangle}
\def\<{\langle}

\newcommand{\veps}{\varepsilon}
\begin{document}
\title{An entropic invariant for 2D gapped quantum phases}
\author{Kohtaro Kato}
\affiliation{Institute for Quantum Information and Matter, California Institute of Technology, Pasadena, CA, USA}
\author{Pieter Naaijkens}
\affiliation{JARA Institute for Quantum Information, RWTH Aachen University, Germany}
\affiliation{Facultad de Ciencias Matem{\'a}ticas, Universidad Complutense de Madrid, Spain}
\begin{abstract}
We introduce an entropic quantity for two-dimensional (2D) quantum spin systems to characterize gapped quantum phases modeled by local commuting projector code Hamiltonians. 
The definition is based on a recently introduced specific operator algebra defined on an annular region, which encodes the superselection sectors of the model. 
The quantity is calculable from local properties, and it is invariant under any constant-depth local quantum circuit, and thus an indicator of gapped quantum spin-liquids. 
We explicitly calculate the quantity for Kitaev's quantum double models, and show that the value is exactly same as the topological entanglement entropy (TEE) of the models.
Our method circumvents some of the problems around extracting the TEE, allowing us to prove invariance under constant-depth quantum circuits. 
\end{abstract}
\maketitle
\section{Introduction}
A gapped quantum phase is an equivalence class of the ground states of gapped local Hamiltonians which are connected by an adiabatic path~\cite{ChenGW10}.
Topologically ordered phases~\cite{Wen89} are gapped quantum phases which exhibit topology-dependent ground state degeneracy and anyonic excitations obeying fractional or non-abelian statistics. 
Ground states in topologically ordered phases do not break any symmetry of the system, and therefore these phases cannot be characterized by the conventional methods of symmetry-breaking and local order parameters. 
Moreover, the characteristic topological properties are robust against any local perturbations. It is proposed to utilize these properties to build a fault-tolerant quantum memory/computer~\cite{Kitaev03,Freedman03}. 
For these reasons, characterizing and classifying topologically ordered phases has attracted a great interest in quantum many-body physics and quantum information science.

A characteristic feature of states in topologically ordered phases is the existence of large-scale multipartite correlations. This is in contrast to two-point correlations which decay exponentially with distance for all gapped systems with sufficiently local interactions~\cite{HastingsK06,NachtergaeleS06}.  
The large-scale correlations are characterized by (dressed) closed-string operators which have constant expectation values for arbitrary loops~\cite{LevinW06,PhysRevB.72.045141}. 
However, it is a demanding task in general to find these non-local operators for given gapped models. 
Levin and Wen proposed a way to avoid this problem: quantifying a contribution of these non-local operators by looking the conditional mutual information, a linear combination of the information-theoretical entropy of the reduced states of certain regions~\cite{LevinW06}. 
The conditional mutual information is purely determined by local reduced states of the ground state wave function, and it is indeed possible to calculate analytically or numerically for various systems~\cite{LevinW06,Zhang11,Isakov11,Jiang12}.  
At the same time the entropic contribution is shown to be equivalent to the so-called topological entanglement entropy proposed by Kitaev and Preskill~\cite{KitaevP06} and Levin and Wen~\cite{LevinW06} (see also~\cite{PhysRevA.71.022315}), which is defined as a non-trivial sub-leading term of the area law of the entanglement entropy. 
The topological entanglement entropy is also shown to be equal to the logarithm of the total quantum dimension in specific models~\cite{LevinW06,KitaevP06}, which is solely determined by the corresponding anyon model of the phase. 
According to these results, the conditional mutual information (or more generally, the tripartite information~\cite{KitaevP06}) thus provides an extraction method for the topological entanglement entropy, and a non-zero value has been regarded as a good signature of topological order.  

However, the equivalence between the topological entanglement entropy (in the sense of the constant term or the conditional mutual information) and (the logarithm of) the total quantum dimension breaks down in some gapped systems.
It has been shown that there exist a ground state in the topologically trivial phase that has non-zero constant term in the area law (sometimes called ``spurious'' topological entanglement entropy)~\cite{PhysRevB.94.075151,Williamson18}. 
These counterexamples have some exotic boundary state at the boundaries of particular subregions, which have a non-trivial symmetry-protected topological order (SPT) characterized by e.g., string order parameters~\cite{kennedy1992}. 
Therefore, the conditional mutual information is not always a good indicator of topological orders, and we need additional conditions to guarantee the relation to the total quantum dimension. 
One possible approach to attack this problem is understanding when this phenomenon happens. 
It may be true that the spurious topological entanglement entropy only arises when the boundary has non-trivial SPT, and the value is not stable under deformations of the regions or some local perturbations.  
However, it has been not yet completely understood under what conditions topologically trivial states can have non-trivial entropic contribution to the conditional mutual information. 

In this paper, we take a different approach, by finding another quantity to quantify the entropic contribution of the characteristic non-local correlations only arising in non-trivial topologically ordered phases. 
We require that the quantity is an invariant of gapped phases, and that it vanishes if the system is in the topologically trivial phase. We also require that the quantity is locally calculable, in the sense that it only depends on local properties of the ground state (although it may be intractable or computationally expensive to calculate). 
Moreover, it would be desirable that the quantity represents the genuinely topological part of the conditional mutual information, in the sense that it coincides with the logarithm of the total quantum dimension for known models. 
To find such a quantity, we take an algebraic approach which is motivated by the work of Haah~\cite{Haah16}. 
Haah introduced an algebra of observables supported on an annulus and showed that it has a non-trivial structure (superselection rule) only in topologically ordered phases. 
He constructed an invariant of gapped quantum phases based on the non-trivial algebraic structure which is an analog of the so-called (modular) $S$-matrix (see e.g.~\cite{Verlinde88, Wang} for the definition) characterizing the anyon models behind the topological order. The invariant is defined as an expectation value of a certain product of operators. Here, we consider an entropic function of the reduced state of the ground state to build a connection to the topological entanglement entropy and the conditional mutual information.

To define the entropic quantity, we first identify the algebra of observables $\calE$ that do not create any additional excitations in an annular region.
This algebra includes the algebra introduced in Ref.~\cite{Haah16} as a subalgebra, and the subalgebra decomposes into different components, related to the superselection sectors (or anyon types) of the theory.
To obtain a canonical representation of this algebra on a Hilbert space containing only relevant states, we apply the GNS construction from the theory of $C^*$-algebras to a ground state restricted to this algebra. 
The corresponding GNS Hilbert space can naturally be decomposed into subspaces corresponding to the superselection sectors of the algebra.

Second, we choose the quantum relative entropy $S(\cdot\|\cdot)$  as a particular distance measure and choose a reference state respecting the superselection rule. 
More precisely, for a ground state $|\Omega\rangle$ we can use $\calE$ to define a Hilbert space isomorphic to $\calE |\Omega\>$, and take the completely mixed state $\tau$ on this space as a reference state.
For a given annular region $A$ (used to define $\calE$), we can then trace out the complement to get $\tau_A$, which is simply obtained from the ground state projector of interactions around $A$. 
Our invariant is then given by
\[
\calI(A)_\Omega:=S\left(\rho_A\left\|\tau_A\right.\right)\,,
\]
where $\rho_A$ is the reduced state of the ground state $|\Omega\>\<\Omega|$ on $A$. 
As we will see later, under suitable conditions, this quantity measures the relative dimension of the trivial sector compared to the dimension of the Hilbert space of all sectors. 

We then proceed to show that this is indeed an invariant of gapped phases.
More precisely, we consider constant-depth geometrically local circuits.
These circuits are obtained by applying a constant (in the system size) number of layers, where each layer is given by a tensor product of local unitary operators. 
The unitary evolution corresponding to any gapped path of Hamiltonians can be (approximately) represented by such a circuit~\cite{PhysRevB.72.045141}.
It follows that we can use them to relate the different ground states in the same gapped phase.
Finally, we calculate the invariant for the quantum double models and show the equivalence to the logarithm of the  total quantum dimension. 

Our framework is an extension of that of Haah, and for this reason we have to make the same (or slightly stronger) assumptions as he does.
That is, we assume that the Hamiltonian is of locally commuting projector code (LCPC) type, that the ground states obey the local topological quantum order (LTQO) condition, and that certain ``logical algebras'' are stable under changes of the shape of region that preserve the topology. While the assumptions look strong for general gapped systems, our method is applicable for all models which are in the same phase as at least one fixed-point model satisfying all assumptions.

The structure of this paper is as follows. In Sec.~\ref{sec:setting}, we introduce all assumptions and the operator algebras which we need to define the entropic invariant.
In Sec.~\ref{sec:entropic}, we define the entropic quantity based on these operator algebras and show the invariance under any constant-depth local circuit. 
We also calculate the quantity for the toric code~\cite{Kitaev03}, the simplest quantum double model.
We finally discuss a relation between our quantity and the original topological entanglement entropy in Sec.~\ref{sec:TEE}. 
In the Appendix, we explicitly calculate the quantity for general quantum double models and also discuss the Fibbonacci model, which cannot be described by the quantum double models.
We also recall some background material on the GNS construction and make a comparison to the sector analysis in the thermodynamic limit.

\section{Formal setting and assumptions}\label{sec:setting}
Throughout this paper, we will consider quantum spin systems arranged on a two-dimensional lattice. 
We first introduce some notation.
For simplicity, we will in particular consider a square lattice $\Lambda_L$ of linear size $L$ (hence the total number of sites $N=\calO(L^2)$) composed of $d$-dimensional quantum systems occupying every site, with $d<\infty$. 
We denote the corresponding Hilbert space on $\Lambda_L$ by $\calH$. The Hilbert space associated to all spins in a subregion $A\subset\Lambda_L$ is denoted by $\calH_A$. We call a subregion including all spins within a circle with radius $r$ a disc (or a ball) of size $r$, and denote it by $b(r)$.  An important type of subregions is an \emph{annulus}, which is defined as $b(R)\backslash b(r)$ for $r<R$, where $b(R)$ and $b(r)$ share the same center. We will denote $b(R)^c$, the complement of $b(R)$, by $D_{out}$ and the inner disc $b(r)$ by $D_{in}$ (see Fig.~\ref{fig:Annulus}). 

We say a bounded operator $O\in\calB(\calH)$ has support $A\subset \Lambda_L$, or $O$ is supported on $A$, if $O=O_A\ot I_{A^c}$, where $O_A\in\calB(\calH_A)$ and $I_{A^c}$ is the identity operator on $\calB(\calH_{\Lambda\backslash A})$. We also denote the support of $O$ by $\supp(O)$. We consider a geometrically local Hamiltonian $H$ on $\calH$,
\begin{equation}
H=-\sum_{j}h_j\,,
\end{equation}
such that each $h_j$ is supported on  $b(w)$ with $w>0$ containing the spin $j$ at the center, with $w$ independent of $L$. 
For a region $X\subset\Lambda_L$, we denote $X_+:=\bigcup_{\supp(h_j)\cap X\neq\emptyset}\supp(h_j)$ (Fig.~\ref{fig:Annulus}).

\begin{figure}[htbp]  
\begin{center}
\includegraphics[width=7.5cm]{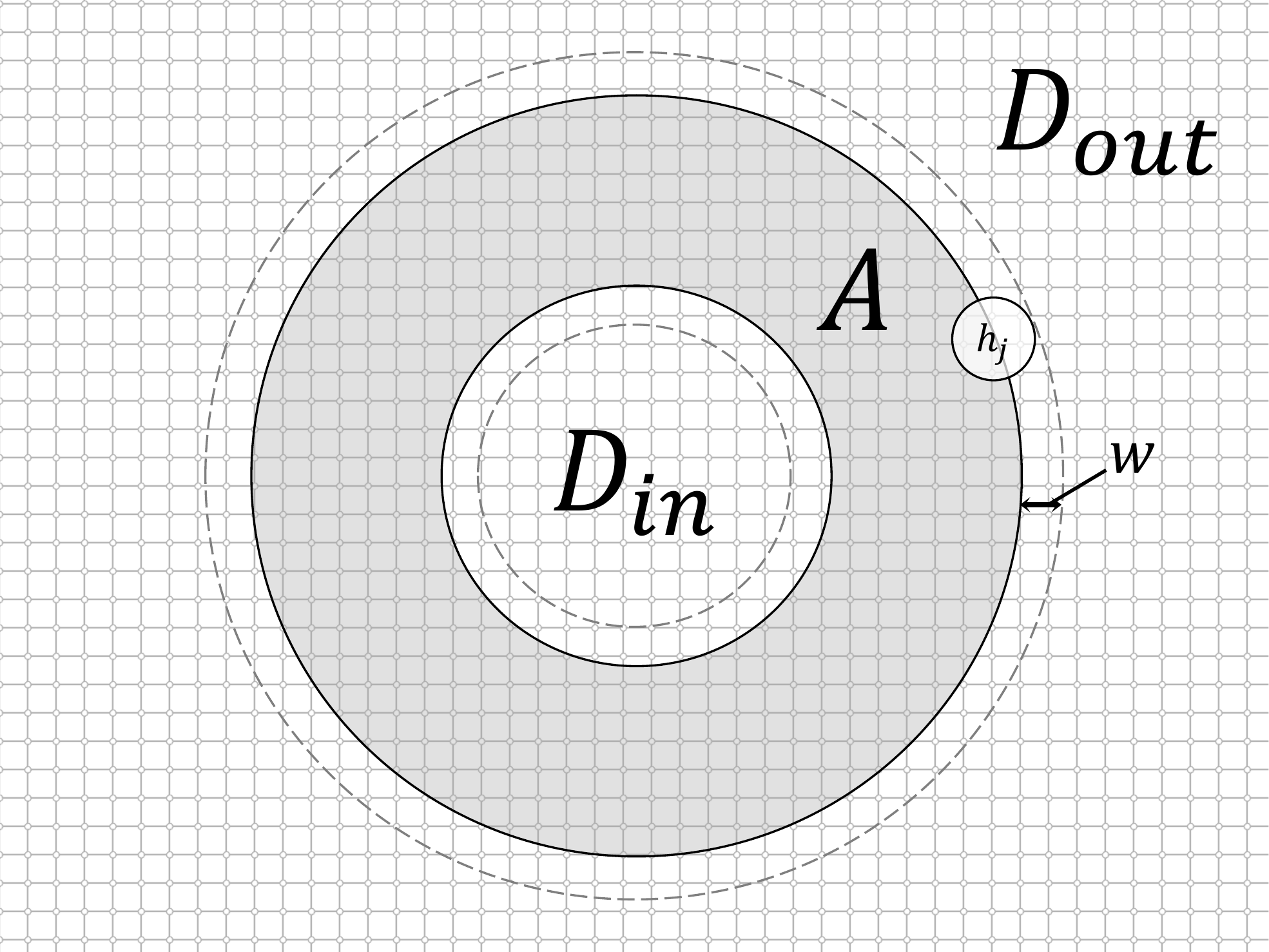}
\end{center}
\vspace{-4mm}
\caption{An annulus $A$ defined by the solid line boundaries on a square lattice (gray region). $D_{out}$ is the outer side of $A$ and $D_{in}$ is the inner side of $A$.  
The dotted line represents the boundary of larger annulus $A_+$, which includes all supports of $h_j$s overlapping with $A$. }
\label{fig:Annulus}
\end{figure}

\subsection{Assumptions on the Hamiltonian}
We assume that the Hamiltonian is a local commuting projector code (LCPC), that is, every $h_j$ is a projector $h_j^2=h_j = h_j^\dagger$ satisfying $[h_j,h_{j'}]=0$ for any $j'$. 
We further assume that it is frustration-free in the sense that every $h_j$ satisfies
\begin{equation}\label{assum:ff}
h_j|\psi\>=|\psi\>\,
\end{equation}
for any ground state $|\psi\>$ of $H$.
Hence the ground states minimize the energy of each term in the Hamiltonian individually.
Kitaev's quantum double models~\cite{Kitaev03} (including the famous toric code model), and Levin-Wen models~\cite{LevinW05} (which describe a wide variety of non-chiral 2D topologically ordered phases) are examples satisfying these conditions. 

We also require an additional condition on the Hamiltonian, called the \emph{local topological order condition} (LTQO)~\cite{MichalakisZ13}: 
we assume there exists an integer $L^*\leq L$ which scales with $L$ such that the following condition holds. 
\begin{itemize}
\item (LTQO): For any disc $X$ of size $r \leq L^*$, let $O_X$ be any operator acting on $X$ and let $\Pi_{X_+}$ be the projector onto the ground subspace of 
\begin{equation}
H_X=\sum_{\supp(h_j)\subset X_+}h_j\,
\end{equation} 
which is defined on $\calH$ (i.e., it is an operator on the whole lattice). 
Then 
\begin{align}\label{eq:newLTQO}
\Pi_{X_+}O_X\Pi_{X_+}=c(O_X)\Pi_{X_+}\,,
\end{align}
where 
\begin{equation}
c(O_X)=\frac{\Tr(\Pi_{X_+}O_X)}{\Tr\Pi_{X_+}}\,.
\end{equation}
\end{itemize}
Note that because we consider LCPC Hamiltonians we can set $\Delta_0(\ell) = 0$ in the notation of~\cite{MichalakisZ13}. 
LTQO is known to be a sufficient condition for the stability of the spectral gap of general frustration-free local Hamiltonians under local perturbations~\cite{MichalakisZ13}.  LTQO implies the following two additional properties~\cite[Cor.2]{MichalakisZ13}:
\begin{itemize}
\item (TQO-1): 
For any disc $X$ of size $r \leq L^*$, let $O_X$ be any operator acting on $X$. Then
\begin{equation}\label{def:LTQO1}
\Pi O_X\Pi=c(O_X)\Pi
\end{equation}
for $c(O_X)$ defined in the above, where $\Pi$ is the projector onto the ground subspace of $H$. 
\item (TQO-2): 
For any disc $X$ of size $r \leq L^*$, let $O_X$ be any operator acting on $X$ such that $O_X\Pi = 0$. Then
\begin{equation}\label{def:LTQO2}
O_X\Pi_{X_+}=0\,.
\end{equation}
\end{itemize} 
These conditions (which are also called local topological order conditions) are used to show the stability of the spectral gap for LCPC Hamiltonians~\cite{BravyiHM10,BravyiH11}. TQO-1 says that local observables cannot map distinct ground states to each other. 
TQO-2 guarantees that the ground subspace of local region is consistent with that of the whole system. Note that TQO-1 and TQO-2 implies~\cite[Cor.1]{BravyiH11}
\begin{equation}
\Pi_{(X_+)_+}O_X\Pi_{(X_+)_+}=c(O_X)\Pi_{(X_+)_+}\,,
\end{equation}
which is slightly weaker than LTQO condition~\eqref{eq:newLTQO}.

To understand the meaning of LTQO, the following equivalent condition will be useful~\cite[Cor.3]{MichalakisZ13}:
\begin{itemize}
\item (LTQO'): Suppose $|\Omega\>$ is a ground state of $H$ and $X$ is a disc of size $r\leq L^*$. For any $|\phi\>$ such that 
$h_j|\phi\>=|\phi\>$ for all $h_j$ such that $\supp(h_j)\subset X_+$, 
\begin{equation}\label{eq:ltqo22}
\Tr_{X^c}|\phi\>\<\phi|=\Tr_{X^c}|\Omega\>\<\Omega|\,.
\end{equation}
\end{itemize}

Perhaps the best known example of a model that satisfies these assumptions is Kitaev's toric (surface) code~\cite{Kitaev03}.
We will use the example of the toric code throughout this paper to illustrate the new definitions. 
\bex The toric code  is defined on a square lattice on a torus, where a site with local dimension $d=2$ is located on each edge. 
The Hamiltonian is given by
\begin{align}
H&=-\sum_{v} \frac{1}{2}\left(I+A_v\right)-\sum_{p}\frac{1}{2}\left(I+B_p\right)\\
&\equiv-\sum_jh_j\,,
\end{align}
where $A_v=X_{v_1}X_{v_2}X_{v_3}X_{v_4}$ around vertex $v$ (times identities on all other sites) and $B_p=Z_{p_1}Z_{p_2}Z_{p_3}Z_{p_4}$ around plaquette $p$ (see Fig.~\ref{fig:toric1}). 
Here $X$ and $Z$ are the usual Pauli matrices. 
It is easy to check that all terms in the Hamiltonian are projectors and  mutually commute. 

\begin{figure}[htbp] 
\begin{center}
\includegraphics[width=6.0cm]{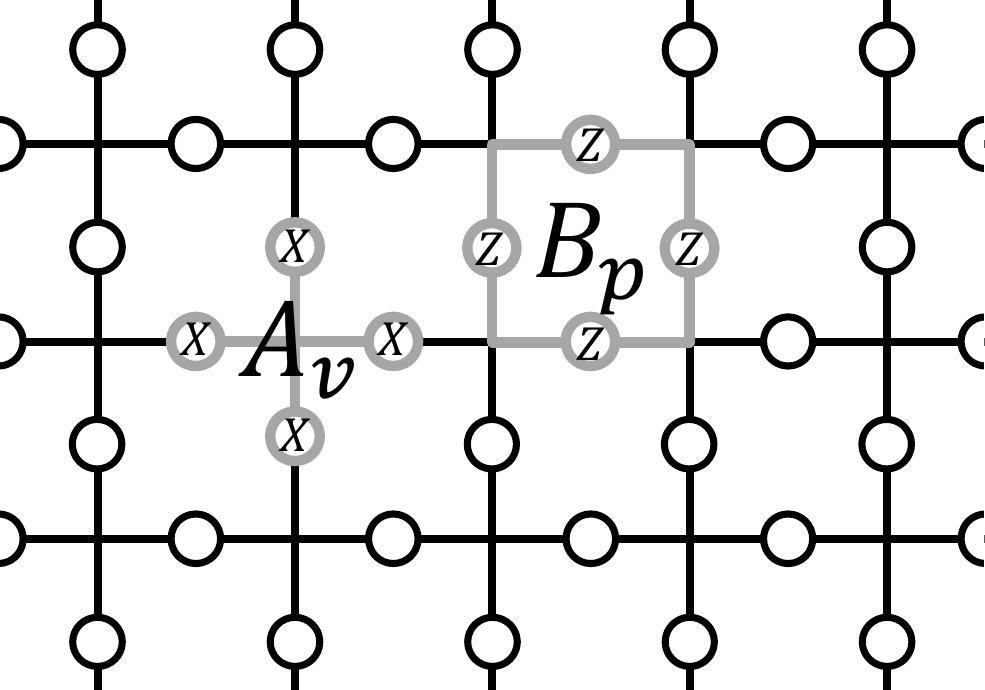}
\end{center}
\vspace{-4mm}
\caption{The interaction terms of the toric code Hamiltonian defined on a square lattice. Each $A_v$ acts on four sites around vertex $v$ and $B_p$ acts on four sites around plaquette $p$.  }
\label{fig:toric1}
\end{figure}

One characteristic feature of ground states of the toric code model is invariance under the actions of closed-string (loop) operators. 
For any path $C$ on the lattice, we define a $Z$-string operator $W_Z(C)$ as a tensor product of Pauli $Z$ operators acting on all spins along $C$. 
In the same way, we can define an $X$-string operator $W_X({\tilde C})$ along a string ${\tilde C}$ on the dual lattice (dual string).  
One can freely deform $W_Z(C)$ ($W_X({\tilde C})$) by applying $A_v$ $(B_p)$ operators neighboring the string, since a product of two identical Pauli operators is the identity. 
When $C$ is a contractible closed string on the manifold on which the lattice is defined, $W_Z(C)$ can be written as a product of all $B_p$ operators supported within the region enclosed by the loop. 
Therefore, any ground state of the toric code model is invariant under the actions of these $Z$-string operators (a similar relation holds for $X$-string operators on dual loops). 

Excitations are created by operators $W_Z(C)$ (or $W_X(\widetilde{C})$) for \emph{open} paths $C$ ($\widetilde{C})$.
Indeed, it is easy to see that $\{A_v, W_Z(C)\} = 0$ if the vertex $v$ is based at one of the endpoints of $C$, and the operators commute otherwise.
Hence if $h_i |\psi\> = |\psi\>$, where $h_i$ is the term containing $A_v$ for the endpoint of $C$, then $h_i W_Z(C) |\psi\> = 0$, and hence it is an excited state.
This state can be understood to have a pair of anyons located at the endpoints of $C$.
If there are no other excitations, the state does not depend on the path $C$, only on its endpoints.
The argument is the same as for the closed loop case.
The case of dual paths is completely analogous, only there the endpoints are located on the plaquettes.
These localized anyons on vertexes/plaquettes are labeled by elements of a finite set $\calL$ (charges, or superselection sectors), which always includes the vacuum (no excitation) denoted  by $1$. An excitation on a vertex is labeled by $e$, and an excitation on a plaquette is labeled by $m$. A pair of $e$ and $m$ on neighboring vertex and plaquette can be treated as another charge labeled by $\veps$. $\calL=\{1,e,m,\veps\}$ contains all possible types of excitations in the toric code model. 
\eex

\subsection{Logical algebras and sectors}\label{sec:logicalalgebra}
The central objects in this paper are operator algebras defined for an annular region $A\subset \Lambda_L$. We restrict the size of the annulus $R$ to $R\leq L^*$, where $L^*$ is defined as in the LTQO condition.
In quantum error correction theory, an operator is called a logical operator if it acts non-trivially on the ground subspace (the code subspace) while commuting with all interaction terms of the Hamiltonian. 
In a similar way, we consider a set of logical operators on $\Lambda_L$ relative to $A$ which we will denote by $\calE$: 
\begin{equation}\label{def:algE}
\calE:=\left\{O\in \calB(\calH_L)\left|[O,h_j]=0  \quad {\rm if }\;\supp(h_j)\subset A_+ \right.\right\}\,.
\end{equation}
Note that $\calE$ is the set of operators that do not create any excitations in the annulus or at the boundary (but may do so outside of $A$). 
Some distinct operators in $\calE$ act identically on the ground states of $H_{A_+}$. 
To get rid of these degeneracies, we factor out $\calE$ by 
\begin{equation}\label{def:algN}
\calN:=\left\{O\in \calE\left| O\Pi_{A_+}=0\right.\right\},
\end{equation}
where $\Pi_{A_+}$ is the projector as in LTQO.
Note that $\calN$ is an additive subgroup of $\calE$ and closed under the product in $\calE$. 
Furthermore, for any $a\in \calE$ and $b\in \calN$, $ab\Pi_{A_+}=a\cdot0=0$ and $ba\Pi_{A_+}=b\Pi_{A_+}a=0\cdot a=0$, since $[a, \Pi_{A_+}]=0$ by definition. 
Hence $\calN$ is a two-sided ideal of $\calE$ and $\calE/\calN$ is a $C^*$-algebra. 
We note that since we are in finite dimensions, a $C^*$-algebra is just a direct sum of matrix algebras, or alternatively, an algebra of block-diagonal matrices.

The effect of dividing out $\calN$ is that we are left with an algebra acting faithfully on the set of states that look like the ground state on $A_+$.
More precisely, suppose that $A |\psi\> = B |\psi\>$ for some $A,B \in \calE$ and all states $|\psi\>$ that reduce to a ground state of $H_{A+}$ on $A_+$.
Then it follows that $(A-B) \Pi_{A_+} = 0$, and hence $[A] = [B]$ in $\calE/\calN$.

In Ref.~\cite{Haah16}, Haah introduced charges (types of particles) within the hole ($D_{in}$) by considering logical operators supported on the annulus $A$, which generate a subalgebra of $\calE/\calN$ in our notation. 
Let us denote a set of logical operators on $A$ by
\begin{equation}\label{def:algA}
\calA:=\left\{O\in \calE\left|\, \supp(O)\subset A\right.\right\}\,
\end{equation}
and factor it out by $\calN_A:=\calN\cap \calA$. The quotient $\calA/\calN_A$ is then a $C^*$-algebra in the same way as $\calE/\calN$. 
Intuitively, $\calA/\calN_A$ is the algebra of ribbon-like loop operators (Wilson loop operators). These quotient algebras faithfully represent the actions onto the ground subspace of $H_A$.  Actually, we have $\calA/\calN_A\cong \Pi_{A_+}\calA\Pi_{A_+}$ via an isomorphism $[O]\mapsto \Pi_{A_+}O\Pi_{A_+}$.

We will now show that the algebra of logical operators $\calE/\calN$ is isomorphic to a full matrix algebra on a finite-dimensional Hilbert space.
Remember that any finite-dimensional $C^*$-algebra can be decomposed into direct sum of matrix algebras~\cite[Thm. I.11.2]{TakesakiI}. 
First note that $\calZ(\calE)$, the center of $\calE$, is generated by $\{h_j|\supp(h_j)\cap A\neq \emptyset\}$, together with the identity.
Because the operators in $\calE$ that are supported outside of $A_+$ generate a full matrix algebra (which has trivial center), it is enough to consider only algebras supported on $A_+$.
Let $M_k(\mathbb{C})$ be the algebra of all such operators, and choose a projector $h_1$ from the Hamiltonian which is supported in $A_+$.
Then the commutant of $h_1$ in $M_k(\mathbb{C})$ is isomorphic to $h_1 M_k(\mathbb{C}) h_1 \oplus (1-h_1) M_k(\mathbb{C}) (1-h_1)$.
Continuing inductively with the other projections $h_2, h_3, \ldots$, using that they mutually commute, we can find $\calE \cap M_k(\mathbb{C})$ and find that its center is indeed generated by the $h_i$ and the identity. 
Now note that 
\begin{equation}
(1-h_j)\Pi_{A_+}=0\;
\end{equation}
for any $h_j$ such that $\supp(h_j)\subset A_+$. It follows that all elements in $\calZ(\calE)$ are in the equivalence class $[1]\in \calE/\calN$. This implies $\calE/\calN$ has trivial center, and therefore $\calE/\calN$ is isomorphic to a full matrix algebra.  

However, $\calA/\calN_A$ may have non-trivial center and we can decompose it into a direct sum of ``superselection sectors''
\begin{equation}\label{eq:SSRdecoA}
\calA/\calN_A=\bigoplus_{a\in\calL}P_a(\calA/\calN_A)P_a\,,
\end{equation}
where $\calL$ is a finite label set and $\{P_a\}$ are the orthogonal projections satisfying $\bigoplus_aP_a=1_{\calA/\calN_A}$. Note that $\calA/\calN_A$ is naturally embedded in $\calE/\calN$ as a subalgebra, since $\calN_A\subset\calN$.  Haah identified the possible charges in $D_{in}$ as labels $\{a\}$ of these sectors. The projectors $P_a$ are then (the equivalence class of) projective measurement operators which measure the total charge that $D_{in}$ has. The label set is finite, and there always is a distinctive label denoted by ``$1$'' such that $P_1|\Omega\>=|\Omega\>$  for any ground state $|\Omega\>$ of $H$. See Ref.~\cite{Haah16} for more details. 

\bex \label{ex:toric E/N}
(Toric code) The algebra $\calA/\calN_A$ has been explicitly calculated for the toric code in Ref.~\cite{Haah16}. In our notation,
\begin{align}
\calA/\calN_A&=\Span\left\{[I], [W_Z(C)], \left[W_X({\tilde C})\right], \left[W_Z(C)W_X({\tilde C})\right]\right\}\\
&=\bigoplus_{a\in\calL}c_aP_a\,, \quad c_a\in{\mathbb C},
\end{align}
where $C$ $({\tilde C})$ is a (dual) loop operator wrapping the annulus once, and $a=1,e,m,\veps$. Since the path operators square to the identity, it is easy to see that the span indeed defines a $C^*$-algebra. 
The orthogonal projectors $P_a$ are $\frac{1}{4}\left([I]\pm \left[W_Z\left(C\right)\right]\right)([I]\pm[W_X({\tilde C})])$ where the signs are determined by the charges.

To specify the set $\calE$, first recall that any $O\in B(\calH)$ can be expressed in a product Pauli basis as
\begin{equation}
O=\sum_{i_1,...,i_N}c_{i_1...i_N}(X^{i_1^1}\ot\ldots\ot X^{i_N^1})(Z^{i_1^2}\ot\ldots\ot Z^{i_N^2})\,
\end{equation}
with $i_k=(i_k^1,i_k^2)\in\{0,1\}\times\{0,1\}$ and $\sigma_{i}=X^{i^1}Z^{i^2}$. It is clear that $[O, A_v]=0$ if and only if $[Z^{i_1^2}\ot\ldots\ot Z^{i_N^2}, A_v]=0$ for all $(i_1,...,i_N)$ such that $c_{i_1...i_N}\neq0$, since otherwise all nonzero terms are linearly independent and do not vanish. 
We call an operator like $Z^{i_1^2}\ot\ldots\ot Z^{i_N^2}$ a \emph{pattern} of $Z$. The same argument holds for $B_p$ and patterns of $X$. 
Therefore it holds that
\begin{equation}
[O,h_j]=0 \Leftrightarrow O\in \Span\left\{(X^{i_1^1}\ot\ldots\ot X^{i_N^1})(Z^{i_1^2}\ot\ldots\ot Z^{i_N^2}) \left| \left[X^{i_1^1}\ot\ldots\ot X^{i_N^1},h_j\right]=\left[Z^{i_1^2}\ot\ldots\ot Z^{i_N^2},h_j\right]=0\right.\right\}\,,
\end{equation}
i.e., $O\in \calE$ if and only if $O$ is in the span of patterns (and their products) of $X$ and $Z$ which commute with all $h_j$ with support overlapping with $A$. 
We can always represent these patterns by $X$- and $Z$-strings (or loops) with no endpoints in and around $A$. 
These string operators or loop operators generating $\calE$ can be classified as $(i)$ loops (no endpoints), $(ii)$ strings with both endpoints in $D_{in}$, or $D_{out}$ and $(iii)$ strings connecting $D_{in}$ and $D_{out}$. 

By dividing $\calE$ by $\calN$, any two elements which can be transformed from one to the other by applying $A_v$ or $B_p$ with support overlapping $A$ are the same. 
Loop operators supported on $A$ are swiped out from the annulus by applying these vertex or plaquette operators, and general loop operators are products of these. 
The representatives of the generators of $\calE/\calN$ are then classified as $(i')$ strings and loops supported either in $D_{in}$ or $D_{out}$ and $(ii')$ strings connecting  $D_{in}$ and $D_{out}$.
\eex

The decomposition in Eq.~\eqref{eq:SSRdecoA} (and also Eq.~\eqref{def:algE}) depends on the choice of $A$ in general. However, we expect that our definition of charges captures a universal property of the model, in the sense that the set of labels (or, equivalently, the number of summands in the decomposition) is preserved by deformations of the region, at least if $A$  is large enough and keeps the topology. 
Moreover, it is natural to assume that the ribbon operators in topologically ordered phases generate an algebra which only depends on the topology of the support region, not the shape or the size of it. 
For a similar reason, a condition called \emph{stable logical algebra condition} has been introduced in Ref.~\cite{Haah16}.
We will require a slightly more general condition: 

\begin{itemize}
\item (Uniform stable logical algebra condition): Let $\calA_{t,r}/\calN_{A_{t,r}}$ denote the logical algebra associated to annulus $A_{t,r}$, which is given by $b(r+t/2)\backslash b(r-t/2)$ for some $t$ and $r$. Then, for any $10w\leq t\leq r/10$ and $10w\leq t'\leq {r'}/10$ with $r,r'\leq L^*$, 
\begin{equation}\label{def:stability}
\calA_{t,r}/\calN_{A_{t,r}}\cong \calA_{t',r'}/\calN_{A_{t',r'}}\,.
\end{equation}
\end{itemize}
Compared to Haah's definition, in addition to being able to change the width of the annulus, we also allow changing the radius.
Because of the topological nature of the models we are interested in, we do not expect our \emph{a priori} slightly stronger assumption to limit the class of models our result applies to. 

Note that in the following we will assume all annuli satisfy the restrictions in this condition. 
We can show that our uniform stable logical algebra condition implies Haah's stable logical algebra condition, which in particular requires that the natural inclusion map induces an isomorphism (which is not assumed in our definition).
\begin{prop}\label{prop}
Let $A_1 \subset A_2$ be two annuli and assume the uniform stable logical algebra condition~\eqref{def:stability}.
Then the identity map provides a natural embedding $\iota : \calA_1 \to \calA_2$, which induces an isomorphism $\calA_1/\calN_{A_1} \to \calA_2/\calN_{A_2}$ of the quotient algebras, where we used the notation of Section~\ref{sec:logicalalgebra}.
\end{prop}
The proof can be found in Appendix~\ref{app:opalg}.
There is a useful consequence of this result.
It says that the inclusion map always induces an automorphism.
Hence to verify that a certain model satisfies the uniform stable logical algebra condition, it is enough to check that for suitable inclusions of annuli, the natural inclusion map induces an isomorphism of the quotient algebras.
That is, if this happens to be false, one does not have to search for other potential isomorphisms. 

Using Proposition~\ref{prop} the following corollary follows easily:
\begin{cor} \label{cor:Abelian} Under the assumption of the uniform stable logical algebra condition, $\calA/\calN_A$ is Abelian.
\end{cor}
\begin{proof}
Let us consider three annuli $A_1$, $A_2$ and $A_3$ such that $A_1\sqcup A_2\subset A_3$ (Fig.~\ref{fig:A=AA}).
Without loss of generality, we assume $\calA/\calN_A$ is defined on $A_3$.
From the discussion above, for any $[O], \in\calA/\calN_A$, there exist $O_1$ and $O_2$ with $[O_1]=[O_2] = [O]$ which are supported on $A_1$ and $A_2$, respectively.
Then, for any $[O], [Q]\in \calA/\calN_A$, we can choose a pair of representatives $O_1$ and $Q_2$ with disjoint supports.
Therefore,  
\begin{align}
[O][Q]=(O_1+\calN_A)(Q_2+\calN_A)=(Q_2+\calN_A)(O_1+\calN_A)=[Q][O] 
\end{align}
and $\calA/\calN_A$ is Abelian.
\begin{figure}[htbp]
\begin{center}
\includegraphics[width=8.0cm]{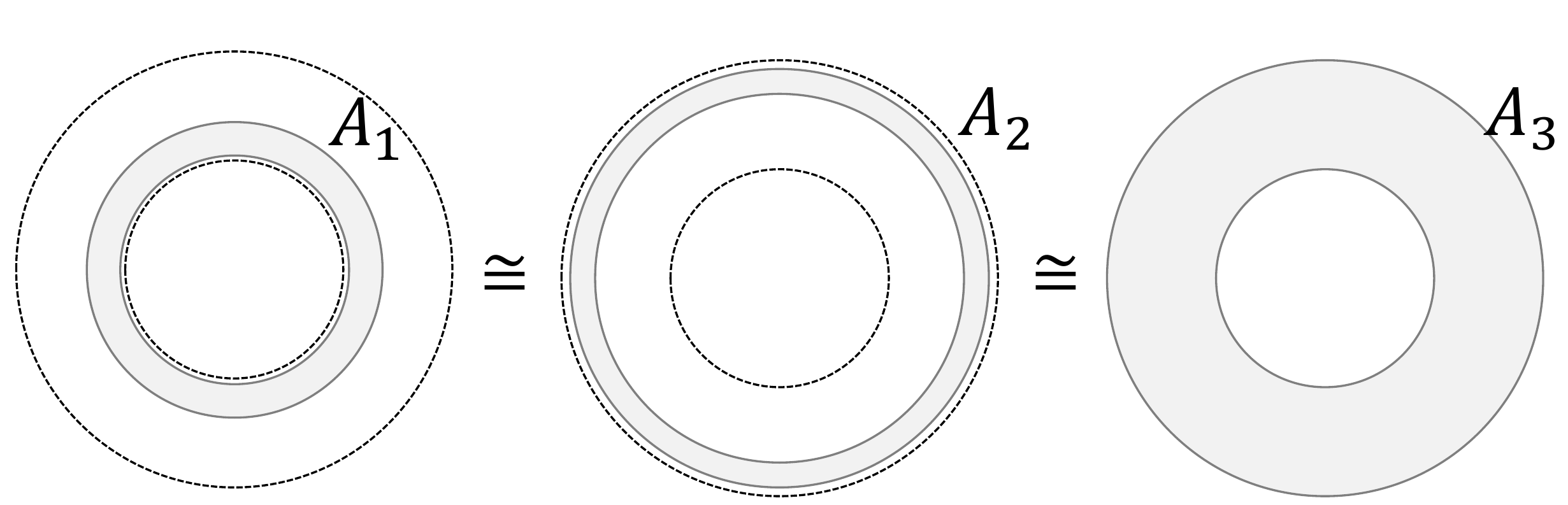}
\end{center}
\vspace{-4mm}
\caption{Annuli $A_1$, $A_2$ and $A_3$ in the proof of Corollary~\ref{cor:Abelian}. Any logical operator on $A_3$ have two disjoint supports $A_1$ and $A_2$ simultaneously, which implies. The isomorphisms from the algebras on smaller regions to the algebra on $A_3$ are established by the natural embedding.}
\label{fig:A=AA}
\end{figure}
\end{proof}

The total charge in $D_{in}$ is not a conserved quantity under the action of  $\calE/\calN$. In other words, there exist operators in $\calE/\calN$ that do not commute with the charge projectors $P_a$, since $\calE/\calN$ has trivial center. 
For example, $\calE/\calN$ contains operators that create a pair of conjugate excitations, one located in $D_{in}$ and one in $D_{out}$. 
These operators change the corresponding charge in $D_{in}$ without making any additional excitation in the annulus. 
The set of logical operators preserving the total charge in $D_{in}$ is given as a subalgebra of $\calE/\calN$:
\begin{align}\label{eq:logicalC}
\calC:=\left(\calZ(\calA/\calN_A)\right)'\cap \calE/\calN=\bigoplus_aP_a(\calE/\calN)P_a\;\equiv\bigoplus_a\calC_a\,.
\end{align}
The equality in the middle follows because $P_a$ are mutually orthogonal projections which generate $\calZ(\calA/\calN_A)$.
Note that we again get a decomposition in terms of the superselection sectors.

\bex {\it (Toric code)} Recall that $\calE/\calN$ is spanned by (the equivalence classes of) $(i')$ string/loop operators supported on either $D_{in}$ or $D_{out}$ whose endpoints are not in and around $A$ and $(ii')$ string operators connecting $D_{in}$ and $D_{out}$ (see Example~\ref{ex:toric E/N}). All non-trivial operators in $(ii')$ do not commute with $\calZ(\calA/\calN_A)$, since they create non-trivial excitations which are detected by some projector $P_a \in \calZ(\calA/\calN_A)$. 
The algebra $\calC$ is thus spanned by operators in $(i')$, and $\calC_a=\calL\vee P_a$ for some finite algebra $\calL$ supported on $A^c$ such that all $P_a$ commute with $\calL$. 
Therefore, different $\calC_a$ are isomorphic each other. 
\eex

From the algebra $\calE$ and the state $|\Omega\>$, we can construct a cyclic representation of $\calE$ on a certain Hilbert state $\calH_\Omega$ in terms of what is called the GNS {\it representation} (see Appendix~\ref{app:opalg} for more details). 
A representation is called \emph{cyclic} if there is a vector such that by acting on this vector via the representation we can span the whole Hilbert space. 
More intuitively, the GNS Hilbert space $\calH_\Omega$ is simply equivalent to the ground subspace of $H_{A_+}$:
\begin{equation}\label{eq:homega}
	\mathcal{H}_\Omega \cong \calE|\Omega\>=\left\{|\psi\>\in\calH\,\left|\;  h_j|\psi\>=|\psi\> \quad {\rm if }\;\supp(h_j)\subset A_+ \right.\right\}\,.
\end{equation}
We refer to Appendix~\ref{app:opalg}, where the equality in equation~\eqref{eq:homega} is proven. 
For clarity we will write this GNS representation as $\pi^\Omega$, which is, in particular, 
an irreducible representation (see Appendix~\ref{app:opalg}). From the GNS construction it is also clear that the space does not depend on the specific choice of ground state $|\Omega\rangle$. 
In the rest of this paper, we will equate $\calH_\Omega$ with $\calE|\Omega\>$, a subspace of $\calH$ defined on $\Lambda_L$.
\\

We can obtain representations of related algebras from $\pi^\Omega$ in natural ways.
For example, $\calN$ is in the kernel of $\pi^\Omega$ (Lemma~\ref{lem:nkernel}), hence $\pi^\Omega$ induces a representation of $\calE/\calN$.
One can also obtain a representation of subalgebras of $\calE/\calN$ by restricting the GNS representation $\pi^\Omega$. 
In particular, we can restrict $\pi^\Omega$ to $\calC \subset \calE/\calN$, the algebra of all logical operators preserving the total charge of $D_{in}$.
This representation is \emph{reducible}, while the representation of $\calE$ is irreducible as we have seen above. 
The GNS Hilbert space $\calH_\Omega$ can be decomposed using $\bigoplus_aP_a= [I] \in\calE/\calN$, where $P_a$ are as in Eq.~\eqref{eq:logicalC}:
\begin{equation}
\calH_\Omega=\bigoplus_a\pi^\Omega(P_a)\calH_\Omega\equiv\bigoplus_a\calH_\Omega^a\,.
\end{equation}
Each sector $\calH_\Omega^a$ is invariant under the action of $\pi^\Omega(\calC)$, and therefore $\pi^\Omega$ of $\calC$ has non-trivial subrepresentations for each $\calH_\Omega^a$. 
For instance, Eq.~\eqref{eq:logicalC} implies that
\begin{align}
\pi^\Omega(\calC)|\Omega\>&\cong\bigoplus_aP_a(\calE/\calN)P_a|\Omega\>\\
&= P_1(\calE/\calN)|\Omega\>\cong\calH_\Omega^1\,,\label{eq:HOmega1}
\end{align}
since  $P_a|\Omega\>=\delta_{a1}|\Omega\>$. Note that here we identify the action of $[A]\in\calE/\calN$ by $[A]|\Omega\>=A|\Omega\>$, which is well-defined.  
We will call the subrepresentation of $\calC$ on $\calH_\Omega^1$ the vacuum representation. 

\bex 
(Toric code) Recall that $\calE$ for the toric code is spanned by $X$ and $Z$-string or loop operators with no endpoints around $A$. Loop operators act trivially on a ground state $|\Omega\>$, 
and commute with any string operators up to a phase ($\{X,Z\}=0$ on the same site). Therefore a basis of $\calE|\Omega\>$ can be constructed by applying only open string operators to $|\Omega\>$.  
Each element of this basis is specified by the pattern of excitations outside of $A$, since two products of string operators sharing the same endpoints differ only by a phase. 
When the numbers of $e$ and $m$ anyons in a basis element are both even in $D_{in}$ (equivalently in $D_{out}$), then the basis element is in the vacuum sector $\calH_\Omega^1$.  
A basis of $\calH_\Omega^a$ for $a\neq1$ is then constructed by applying one string operator creating a pair of anyons with the charge $a$ in $D_{in}$ and $D_{out}$ to the vacuum sector. 
The string operators are unitary and induce an isomorphism such that $\calH_\Omega^a\cong\calH_\Omega^1$. Therefore, $\calH_\Omega=\bigoplus_a\calH_\Omega^a$ is a direct sum of four isomorphic orthogonal sectors. 
\eex

\section{An entropic invariant of 2D gapped phases}\label{sec:entropic}
When two ground states of gapped Hamiltonians are connected via an adiabatic evolution without closing the gap for all system sizes, they are said to be in the same gapped quantum phase (see e.g. Ref.~\cite{ChenGW10} for more precise definition).   
By using the technique of quasi-adiabatic continuation~\cite{PhysRevB.72.045141}, one can show that this definition of phase is equivalent to considering a particular unitary evolution mapping one to the other.  
Importantly, this unitary evolution is generated by quasi-local Hamiltonians and can be simulated by a constant-depth local quantum circuit with a constant error. A constant-depth local 
(quantum) circuit is defined as a unitary which can be represented as a product of unitaries 
\begin{equation}
W=W^{(1)}W^{(2)}\ldots W^{(M)}\,,
\end{equation}
where $M$ is a constant independent of the system size and each $W^{(i)}=\bigotimes_lW_l^{(i)}$ is a tensor product of unitaries acting on constant-size disjoint sets of neighboring sites.  
Any constant-depth local quantum circuit maps a (geometrically) local operator to a local operator with slightly larger support. We say a circuit has range $r$ if the support of a local operator spreads at most distance $r$ from the initial support after the transformation. 
In this paper, we will only consider invariance under constant-depth local circuits as in Ref.~\cite{Haah16}.
Although these transformations only approximate quasi-adiabatic evolutions, we believe that our results can be extended with some additional errors vanishing in appropriate thermodynamic limit (cf.~\cite{BachmannMNS12}). 

Using the assumptions that we have made so far, $\calA/\calN_A$ has been shown to be invariant under any constant-depth local circuits~\cite{Haah16}. 
Therefore, a quantity which only depends on the  algebraic structure of the logical algebras must be an invariant in the same way. More precisely, Haah proves that the algebras are \emph{isomorphic}, so one has to show that the invariant is stable under isomorphisms.

In this section we introduce a new entropic quantity that is invariant under constant-depth local circuits.
We first define the entropic quantity that essentially measures the relative sizes of the superselection sectors compared to the trivial (ground state) sector. 
We then provide a formula to calculate the quantity in terms of the dimensions of the sectors by considering the information convex introduced in Ref.~\cite{Shi18} in our context. 
	We prove that this quantity is stable under constant-depth local circuits, under the assumptions stated earlier.

\subsection{Definition of the Entropic Invariant}
Now we are ready to define the entropic quantity.  
Our strategy is to choose a good reference point in the information convex and quantify the difference to probe the nontrivial structure of $\Sigma(A)$. We choose the reduced state of the completely mixed state on $\calH_\Omega$ as the reference state. 
As a measure of difference of two quantum states, we use the relative entropy:
\begin{equation}
S(\rho\|\sigma):=\Tr\rho\log(\rho-\sigma)\,
\end{equation} 
where the logarithm is in base 2.
The relative entropy is zero if and only if the two states are the same, and positive otherwise.
It is however not a proper ``distance'' since it does not satisfy the triangle inequality. 
This ``quasi-distance'' is frequently used in information theory because of its various useful properties.

We propose the following quantity as a new entropic invariant of gapped phases.
\bdf 
Consider a ground state $|\Omega\>$ of $H$. For $\rho_A=\Tr_{A^c}|\Omega\>\<\Omega|$ on a given annulus $A$, we define 
\begin{equation}\label{def:invariant}
\calI(A)_\Omega:=S\left(\rho_A\left\|\tau_A\right.\right)\,,
\end{equation}
where $\tau_A=\Tr_{A^c}\tau$ is the reduced state on $A$ of the completely mixed state $\tau$ on the Hilbert space $\calH_\Omega$ (recall that we equate $\calH_\Omega$ with $\calE|\Omega\>\subset\calH$).
\edf

Since $\calH_\Omega$ is (isomorphic to) the ground subspace of $H_{A_+}$, the completely mixed state is just given by $\tau\propto\Pi_{A_+}={\tilde \Pi}_{A_+}\ot1_{(A_+)^c}$, where ${\tilde \Pi}_{A_+}$ is a projector of the ground subspace of $H_{A_+}$ restricted to $\calH_{A_+}$.   
Therefore, it is determined from $H_{A_+}$ and locally calculable. However, it might be true that $\calI(A_t)_\Omega$ \emph{does} depend on the choice of the annulus. We distinguish the desirable case in which it is independent from the choice of the annulus.
\bdf\label{def:uniform}
We say $\calI(A_t)_\Omega$ is uniform if it is independent of $t$ for $10w<t<r_{ann}-10w$. 
\edf

The uniform property of $\calI(A)_\Omega$ is related to the stability of $\calE/\calN$ in the sense of the stability of $\calA/\calN_A$ in Eq.~\eqref{def:stability}.   
It might be true that $\calI(A)_\Omega$ is always uniform by the stable logical algebra condition, but unfortunately we do not have a rigorous proof yet. 
A sufficient condition for the uniform property is that the following two properties holds: $i)$ for two annuli $A\subset B$ with different thicknesses, $\tau_A=\Tr_{A^c}{\tilde \tau}_B$, where $\tau$ and ${\tilde \tau}$ are the completely mixed states on the corresponding GNS Hilbert spaces. 
$ii)$ there is a ``recovery'' CPTP-map $\calR$ such that $\calR(\rho_A)=\rho_B$ and $\calR(\tau_A)=\tau_B$. Then, from the joint monotonicity of the relative entropy, we have
\begin{equation}
\calI(B)_\Omega=S(\rho_B\|\tau_B)\geq S(\rho_A\|\tau_A)=\calI(A)_\Omega\geq S(\rho_B\|\tau_B)=\calI(B)_\Omega\,,
\end{equation}
which implies the uniform property of $\calI(A)_\Omega$.  
As we will see in later, the quantum double model satisfies the uniform property.

For models without commuting Hamiltonian terms, it would be reasonable to relax the statement to be only approximately true, up to a controllable error, for sufficiently large annuli. 
As evidence for this conjecture, we remark that $\calI(A)_\Omega$ is uniform for Kitaev's quantum double model (see Appendix~\ref{appendix:QD}). 
However, $\calI(A)_\Omega$ could depend on the annulus for more general models, e.g. in the Levin-Wen models. 
In the next section it will be shown that if $\calI(A)_\Omega$ is uniform, then $\calI(A)_{W\Omega W^\dagger}$ is also uniform for any constant-depth local circuit $W$ (Theorem~\ref{thm:invariance}). 
Therefore, it is enough to show the uniform property for certain ``fixed-point'' wave functions of gapped phases.

By definition $\calI(A)_\Omega$ is nontrivial if and only if $\rho_A\neq\tau_A$.
Note that by Lemma~\ref{lemma:vacRDM}, $\rho_A$ is the reduced state of the completely mixed state restricted to the vacuum sector $\calH_\Omega^1$.
Hence $\calI(A)_\Omega$ is nontrivial if and only if $\calH_\Omega$ has a nontrivial superselection structure.
Moreover, its value is determined by the dimensions of the sectors of the GNS Hilbert spaces as we will see in the next section. 

\subsection{A Formula for the Invariant via Structure of the Information Convex}
The non-trivial structure of $\calA/\calN_A$ implies that states in $\calH_\Omega$ can have different reduced states on $A$, which is not possible if $A$ is a disc (all ground states of $H_{A_+}$ have the same reduced state on the disc by LTQO'~\eqref{eq:ltqo22}).
In more detail, the set of all possible reduced states on $A$ of states in $\calH_\Omega$ constitute a non-trivial convex set: 
\begin{equation}\label{def:infoconv}
\Sigma(A):=\left\{\sigma_A \in \calS(\calH_A) \left|\, \sigma_A =\Tr_{A^c}\sigma, \, \sigma\in\calS(\calH_\Omega)\right.\right\}\,.
\end{equation}
It is equivalent to the one independently introduced in Ref.~\cite{Shi18} called the \emph{information convex}.
States in the information convex cannot be distinguished locally. In other words, they have the same marginals on every disc-like subregion, say $X$, of $A$. 
This is because any $|\phi\>\in\calH_\Omega$ is in the ground subspace of $H_{X_+}$, and therefore LTQO' guarantees that the reduced states on $X$ is the same as that of $|\Omega\>$.   
States in the information convex obey the superselection structure so that there is no coherence between different sectors.
\bthm\label{thm:convex} Any state $\sigma_A$ in $\Sigma(A)$ can be decomposed as a convex combination: 
\begin{equation}
\sigma_A=\bigoplus_{a}p_a\sigma_A^a\,,
\end{equation}
where $p_a=\Tr(P_a\sigma)$ and $\sigma_A^a=P_a\sigma_AP_a/p_a$, and the $P_a$ are as in equation~\eqref{eq:logicalC}. 
\ethm
\begin{proof}
Consider the equivalence class $P_a={\hat P}_a+\calN_A\in \calZ(\calA/\calN_A)$ with a representative projector ${\hat P}_a\in\calA$. ${\hat P}_a$ belongs to ${\hat P}_a+\calN$ in $\calE/\calN$.  
Now, we are going to show that there exists orthogonal projectors ${\hat Q}_a$ such that $\supp({\hat Q}_a)\subset A^c\cap A_+$, ${\hat Q}_a|\psi\>={\hat P}_a|\psi\>$ for any $|\psi\>\in\calH_{\Omega}$ and $\bigoplus_aQ_a=1_{\calH_\Omega}$.  If this is true, we can show that
\begin{align}
\sigma_A&=\Tr_{A^c}\left(\sigma\bigoplus_a{\hat Q}_a\right)\\
&=\bigoplus_a\Tr_{A^c}({\hat Q}_a\sigma {\hat Q}_a)\\
&=\bigoplus_a\Tr_{A^c}({\hat P}_a\sigma {\hat P}_a)\\
&=\bigoplus_a\Tr({\hat P}_a\sigma)\frac{{\hat P}_a\sigma_A {\hat P}_a}{\Tr({\hat P}_a\sigma)}, 
\end{align} 
which completes the proof. 
Indeed, the existence of such ${\hat Q}_a$ can be shown as a consequence of a result in  quantum error correction theory, which can be stated as follows: 
\begin{prop}\label{prop:OAQEC}~\cite{Almheiri2015} Consider a subspace $\calH_C\subset\calH_R\ot\calH_{R^c}$ and an operator $O\in\calB(\calH_C)$. Then, there exists $O_{R^c}$ supported on $R^c$ such that 
\begin{equation}
O_{R^c}|\psi\>=O|\psi\>
\end{equation}
for any $|\psi\rangle\in\calH_C$ if and only if $P_C[O,X_R]P_C=0$ $\forall X_R\in\calB(\calH_R)$, where $P_C$ is the projection onto $\calH_C$. 
\end{prop}

In our case, $R=A\cup (A_+)^c$, $\calH_C=\calH_\Omega$ and $P_C=\Pi_{A_+}$  and $O\in\Pi_{A_+}\calA\Pi_{A_+}$ (see Lemma~\ref{lemma:GNS=spanE}). One can easily check that
\begin{align}
\Pi_{A_+}[{\hat P}_a,X_A\ot Y_{(A_+)^c}]\Pi_{A_+}&=\left[{\hat P}_a\Pi_{A_+}, \Pi_{A_+}X_A\Pi_{A_+}\right]\ot Y_{(A_+)^c}\,
\end{align} 
for any operators $X_A$ on $\calH_A$ and $Y_{(A_+)^c}$ on $\calH_{(A_+)^c}$. $\Pi_{A_+}X_A\Pi_{A_+}$ is supported on $A_+$ and an element of ${\tilde \calA}/\calN_{A_+}$, the logical algebra associated to the slightly larger annulus $A_+$. From the stability assumption~\eqref{def:stability}, ${\hat P}_a+\calN_{A_+}\in\calZ({\tilde \calA}/\calN_{A_+})$ and therefore $\left[{\hat P}_a\Pi_{A_+}, \Pi_{A_+}X_A\Pi_{A_+}\right]=0$. By linearity, this implies $\Pi_{A_+}[{\hat P}_a,Z_{A\cup(A_+)^c}]\Pi_{A_+}=0$ for any $Z_{A\cup(A_+)^c}\in\calB(\calH_{A\cup(A_+)^c})$. 
Therefore there exists ${\hat Q}_a$ supported on $(A\cup (A_+)^c)^c=A^c\cap A_+$ such that ${\hat Q}_a|\psi\>={\hat P}_a|\psi\>$ for any $|\psi\>\in \calH_\Omega$. 
\end{proof}
Theorem~\ref{thm:convex} says any reduced state of $\calH_\Omega$ is decomposed into a probabilistic mixture of state supported on disjoint sectors. Moreover, the reduced state on a particular sector is essentially unique (the same result has been proven for quantum double models~\cite{Shi18}): 
\blm \label{lemma:vacRDM} Under the uniform stable algebra condition, any state $|\psi^a\>\in\calH_\Omega^a$ has the same reduced state
\begin{equation}
\psi^a_A=\rho_A^a\,,
\end{equation}
where $\rho_A^a=\Tr_{A^c}({\hat P}_a\Pi_{A_+})/\Tr({\hat P}_a\Pi_{A_+})$ for ${\hat P}_a\in P_a$. 
\elm
\begin{proof} Choose an element ${\hat P}_a\in P_a\in \calA/\calN_A$. For any operator $O_A$ supported on $A$, it holds that 
\begin{align}
\Tr(O_A|\psi^a\>\<\psi^a|)=\Tr(\Pi_{A_+}{\hat P}_aO_A{\hat P}_a\Pi_{A_+}|\psi^a\>\<\psi^a|)\,.
\end{align}
The operator $\Pi_{A_+}{\hat P}_aO_A{\hat P}_a\Pi_{A_+}$ is supported on $A_+$ and commutes with all $h_i$, and therefore it is an element of $P_a({\tilde \calA}/\calN_{A_+})P_a$, the logical algebra associated to the slightly larger annulus $A_+$. By Corollary~\ref{cor:Abelian}, $P_a({\tilde \calA}/\calN_{A_+})P_a$ is one-dimensional and therefore 
\begin{equation}\label{eq:one-dOA}
\Pi_{A_+}{\hat P}_aO_A{\hat P}_a\Pi_{A_+}+\calN_{A_+}=c(O_A){\hat P}_a\Pi_{A_+}+\calN_{A_+}\,.
\end{equation}
From the stable logical algebra condition and $\calN_{A_+}\subset\calN_A$, Eq.~\ref{eq:one-dOA} implies
\begin{equation}
\Pi_{A_+}{\hat P}_aO_A{\hat P}_a\Pi_{A_+}+\calN_{A}=c(O_A){\hat P}_a\Pi_{A_+}+\calN_{A}\,.
\end{equation}
Therefore we have 
\begin{align}
\Tr(O_A|\psi^a\>\<\psi^a|)=c(O_A)\Tr(|\psi^a\>\<\psi^a|)=c(O_A)\,,
\end{align}
which completes the proof by the definition of the reduced state. 
\end{proof}

By using Theorem~\ref{thm:convex} and Lemma~\ref{lemma:vacRDM} restricting the structure of the states on $\calH_\Omega$, we obtain the following formula for $\calI(A)_\Omega$. 
\bthm\label{thm:InvGNSdim} For any ground state $|\Omega\>$ and annulus $A$, under the uniform stable logical algebra condition it holds that 
\begin{equation} \label{eq:InvGNSdim}
\calI(A)_\Omega=-\log \frac{d^1_\Omega}{d_\Omega}\,
\end{equation}
where $d_\Omega=\dim\calH_\Omega$ and $d^1_\Omega=\dim\calH_\Omega^1$. 
\ethm
\begin{proof} 
Let us denote the projector onto $\calH_\Omega^a$ by $\Pi_a$ and $\dim\calH_\Omega^a$ by $d_\Omega^a$.  
By definition
\begin{align}\label{eq:decotau}
\tau_A=\bigoplus_a\left(\frac{d_\Omega^a}{d_\Omega}\right)\Tr_{A^c}\left(\frac{1}{d_\Omega^a}\Pi_a\right)\equiv\bigoplus_ap_a\rho_A^a\,,
\end{align}
where $p_a=d_\Omega^a/d_\Omega$ and $\rho_A^a=\frac{1}{d_\Omega^a}\Tr_{A^c}\Pi_{a}$. 
Since $\rho^1_A=\rho_A$ by Lemma~\ref{lemma:vacRDM}, we have 
\begin{align}
\calI(A)_\Omega&=S(\rho_A\|\tau_A)\\
&=\Tr\rho_A\log\rho_A-\Tr\rho_A\log(p_1\rho_A)\\
&=-\log p_1\,.
\end{align}
The second line follows since $\rho_A\log(\bigoplus_ap_a\rho_A^a)=\rho_A\log(p_1\rho_A^1)$.
\end{proof}
\begin{remark} The decomposition~\eqref{eq:decotau} implies that the relative entropy in Eq.~\eqref{def:invariant} is equal to the max-relative entropy~\cite{4957651}
\begin{equation}
S_{\max}(\rho_A\|\tau_A):=\inf_{\lambda}\left\{\log\lambda\,|\,\rho_A\leq 2^\lambda\tau_A\right\}\,.
\end{equation} 
\end{remark}
\begin{remark} This theorem suggests a more algebraic definition. Consider the projections $P_i$ projecting on the different sectors, and choose a faithful tracial state $\tau$. Then one can look at the ratios $\tau(P_i)/\tau(P_1)$ comparing the sector $i$ to the vacuum sector. This is somewhat reminiscent of the definition of the Jones index for Type II$_1$ factors in operator algebra~\cite{JonesSunder}. See also Appendix~\ref{app:fibonacci}.
\end{remark}

Theorem~\ref{thm:InvGNSdim} helps to obtain $\calI(A)_\Omega$ without having to explicitly calculate the reduced states of a ground state and the reference state. Actually, the calculation is very simple for the toric code model. 
\bex\label{ex:toric1/4} For the toric code, $\calH_\Omega$ is the direct sum of four isomorphic Hilbert spaces $\calH^a_\Omega$ $(a=1,e,m,\varepsilon)$. Therefore, $d_\Omega=4d_\Omega^1$ and we have
\begin{equation}
\calI(A)_\Omega=-\log\frac{1}{4}=2\,
\end{equation}
for any (sufficiently large) $A$, $\Omega$ and $L$.
\eex

By definition $\calI(A)_\Omega$ reflects the structure of the ground subspace of $H_{A_+}$. 
To claim that $\calI(A)_\Omega$ quantifies some sort of correlations in the annulus, it is desirable that the function only depends on the states, not the Hamiltonian. 
Indeed, $\calI(A)_\Omega$ is independent of the Hamiltonian in the same way as in the case of $S$-matrix defined in Ref.~\cite{Haah16}. In more detail, it is shown that one can construct $\calA/\calN_A$ solely from the ground state by using its connection to the so-called \emph{locally invisible operators}, which are operators whose action onto the ground state cannot be detected by looking at local regions~\cite{Haah16}. In the same way, $\calE/\calN$ can also be constructed from the ground state, and thus $\calI(A)_\Omega$ takes the same value for two Hamiltonians if both Hamiltonians have $|\Omega\>$ as a ground state and satisfy all assumptions. In this sense, we can argue that $\calI(A)_\Omega$ is a quantity associated to states. We emphasize that strictly speaking $\calI(A)_\Omega$ is a function of $\rho_{A_+}$, not $\rho_A$.

\subsection{Invariance under constant-depth local circuits}
In this section, we show that $\calI(A)_\Omega$ is invariant under any constant-depth local circuit. 
We begin with restating Haah's results on the stability of $\calA/\calN_A$, which will be shown to be useful in the proof.
\begin{prop}\label{lemma:Haahstability}
(\cite[Theorem 4.1]{Haah16}) Suppose a state $|\Omega\>$ on a plane of size $>L$ admits a LCPC Hamiltonian of interaction length $w$ satisfying all our assumptions and let $|\Omega\>$ be a ground state of $H$. Let $\calA/\calN_A$ be the logical algebra constructed from $H$, such that $A$ has radius $r_{ann}$ and thickness $t$. Denote $\tilde{\calA}/\tilde{\calN}_{A}$ the logical algebra on the same annulus but constructed from $W^\dagger HW$ for any constant-depth local circuit $W$ of range $r<t$. Then, whenever $1200w<60t<r_{ann}<L$, there exists an isomorphism such that 
\begin{equation}\label{eq:Haahstability}
\calA/\calN_A\cong\tilde{\calA}/\tilde{\calN}_{A}\,.
\end{equation}
\end{prop}
Note that we choose the constants in the theorem in the same way as in Ref.~\cite{Haah16}, and the precise values themselves are not essential.  
As a simple consequence, the logical algebra of the model in the topologically trivial phase (including Bravyi's counterexample exhibiting nontrivial spurious topological entanglement entropy) is always trivial~\cite{Haah16}. 
This fact implies the following corollary.
\bcor Under the same assumptions as in Proposition~\ref{lemma:Haahstability}, $\calI(A)_\Omega=0$ if the system is in the topologically trivial phase.
\ecor
Therefore, $\calI(A)_\Omega$ can be used as an indicator of topologically ordered phases (if the Hamiltonian satisfies all assumptions). 
Our main theorem in this paper is that $\calI(A)_\Omega$ is not only a witness of the existence of a nontrivial topological order, but also an invariant of gapped phases. 
\bthm \label{thm:invariance}
Under the same assumption as in Proposition~\ref{lemma:Haahstability}, it holds that
\begin{equation}\label{eq:thminvariance}
\calI(A_t)_\Omega=\calI(A_{t-r})_{W\Omega W^\dagger}\,.
\end{equation}
Moreover, if $\calI(A)_\Omega$ is uniform, 
\begin{equation}
\calI(A_t)_\Omega=\calI(A_t)_{W\Omega W^\dagger}\,
\end{equation}
for any $1200w+r<60t<r_{ann}$. Hence $\calI(A_t)_{W\Omega W^\dagger}$ is also uniform. 
\ethm
\begin{proof} The second part of the proof is easily derived from Definition~\ref{def:uniform}. 
We show Eq.~\eqref{eq:thminvariance} by using the formula in Theorem~\ref{thm:InvGNSdim}. 
Let us denote $\tilde h_j=W h_j W^\dagger$, which is an interaction term of the new Hamiltonian after the transformation. 
We define the annuli $A$ as $A_t$ and $A'$ as $A_{t-r}$. Then, $\calE$ is mapped to  
\begin{align}
\tilde{\calE}&:=\left\{O\in \calB(\calH_L)\left|[O,{\tilde h}_j]=0  \quad {\rm if }\;\supp(h_j)\subset A_+ \right.\right\}\\
&\;=\left\{O\in \calB(\calH_L)\left|[O,{\tilde h}_j]=0  \quad {\rm if }\;\supp({\tilde h}_j)\subset A'_{+} \right.\right\}\,,
\end{align}
where $A'_+:=\bigcup_{\supp({\tilde h}_j)\cap A'\neq\emptyset}\supp({\tilde h}_j)=\bigcup_{\supp(h_j)\cap A\neq\emptyset}\supp({\tilde h}_j)$. It is easy to check that $\tilde{\calE}=W\calE W^\dagger$ and $\calH_{W\Omega W^\dagger}\cong\calH_\Omega$.  It also holds that $\calN\cong\tilde{\calN}$ for $\tilde{\calN}:=\{O\in\tilde{\calE}| O\Pi_{A'_+}=0\}$. 
Consider a logical algebra $\tilde{\calA}/\tilde{\calN}_{A'}\subset\tilde{\calE}/\tilde{\calN}_{A'}$ associated to $A'$. By the stable logical algebra condition~\eqref{def:stability}, $\calA/\calN_A$ is isomorphic to the logical algebra of $|\Omega\>$ on $A'$, which is isomorphic to $\tilde{\calA}/\tilde{\calN}_{A'}$ from Proposition~\ref{lemma:Haahstability}. Hence we have $\calA/\calN_A\cong\tilde{\calA}/\tilde{\calN}_{A'}$. 
This isomorphism implies $\tilde{\calA}/\tilde{\calN}_{A'}$ has the same superselection structure and corresponding projectors $P'_a$ as $\calA/\calN_A$. 
Hence, we have $P'_a\calH_{W\Omega W^\dagger}\cong P_a\calH_\Omega$ for every $a$, especially $a=1$. Therefore the dimensions of these isomorphic Hilbert spaces are the same. This completes the proof by Theorem~\ref{thm:InvGNSdim}.
\end{proof}
A key point of the proof is that $\calI(A)_\Omega$ only depends on the ratio of dimensions of the GNS representations (Theorem~\ref{thm:InvGNSdim}).
This is one of the reasons why we choose $\tau_A$ as the reference state. One can use other reference states/measures, while the invariance under constant-depth circuit is not guaranteed in general.  
For instance, Refs.~\cite{PhysRevB.97.144106,Shi18} propose to use the entropy difference as the measure and the maximum entropy state in $\Sigma(A)$ as the reference state.
However, the invariance of such a quantity is not clear. 
Moreover, our choice of the reference state implies that $\calI(A)_\Omega$ is equivalent to the topological entanglement entropy, at least for quantum double models.  

A simple but non-trivial corollary of Theorem~\ref{thm:invariance} is that $\calI(A)_\Omega$ is a universal quantity of the topologically ordered phase of the toric code model ($\mathbb{Z}_2$-topological order).
	\bcor For any ground state in $\mathbb{Z}_2$-topological order, 
\begin{equation}
\calI(A)_\Omega=2. 
\end{equation}
\ecor
This is because the toric code model is known to satisfy all the assumptions~\cite{MichalakisZ13} and the uniform property.

\section{Relation to topological entanglement entropy}\label{sec:TEE}
Ground states of gapped local Hamiltonians are believed to obey an area law: the von Neumann entropy $S(\rho):=-\Tr\rho\log\rho$ of the reduced state of a ground state for a region $A$ scales as
\begin{equation}\label{eq:arealaw}
S(\rho_A)=\alpha|\partial A|-n_A\gamma+o(1)\,,
\end{equation}
where $\alpha$ is a constant depends on the Hamiltonian, $\gamma$ (or -$\gamma$) is called the topological entanglement entropy and $n_A$ is the number of disconnected  boundaries of $A$~\footnote{More generally, there are additional constant terms which depends on the shape of the corners of the region.}. The $o(1)$ term comprises correction terms vanishing in the limit $|A|\to \infty$.
The area law~\eqref{eq:arealaw} can be verified analytically in certain exactly solvable models such as the quantum double model or the Levin-Wen models~\cite{PhysRevA.71.022315,Flammia09,LevinW06}. 
It is also verified numerically in other gapped models~\cite{Zhang11,Isakov11,Jiang12}. 

Assuming an area law as above holds, one way to obtain the topological entanglement entropy is by taking a suitable linear combination of entropies of subregions. 
For an annulus, consider a tripartition as in Fig.~\ref{fig:TEE1} and define the conditional mutual information for the partition:
\begin{equation}
I(X:Z|Y)_\rho:=S(XY)_\rho+S(YZ)_\rho-S(Y)_\rho-S(XYZ)_\rho\,,
\end{equation}
where $S(A)_\rho:=S(\rho_A)$. By inserting the area law~\eqref{eq:arealaw}, the boundary terms cancel out and
\begin{equation}
I(X:Z|Y)_\rho=2\gamma+o(1)\,.
\end{equation}
Moreover, when the area law~\eqref{eq:arealaw} is exactly saturated and the $o(1)$ term vanishes, it is equivalent to (\emph{i}) the relative entropy distance from the set of all local Gibbs state, and (\emph{ii}) the asymptotically optimal rate of certain secret sharing  protocol~\cite{KatoFM16}.  
An annulus is not the only type of region we can choose: one could for example use a tripartite disc or a more complicated region to extract $\gamma$, as long as taking suitable combinations of the entanglement entropies cancel out the area terms in the area law.

\begin{figure}[htbp] 
\begin{center}
\includegraphics[width=5.5cm]{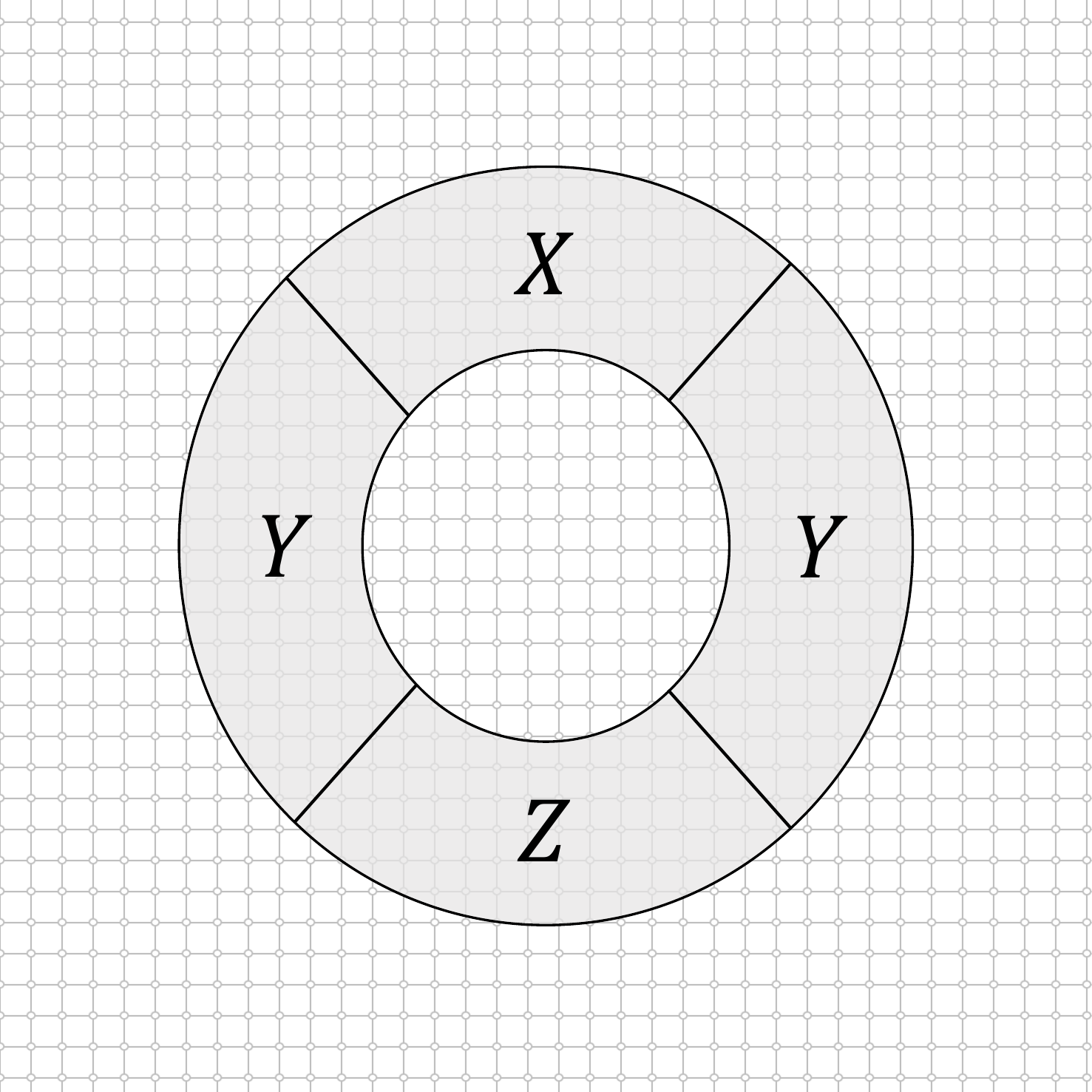}
\end{center}
\vspace{-4mm}
\caption{A tripartition of an annulus for the calculation of the topological entanglement entropy. In this case, the topological entanglement entropy is equivalent to the conditional mutual information $I(X:Z|Y)_\rho$, under certain assumptions. }
\label{fig:TEE1}
\end{figure}

The topological entanglement entropy is argued to be a universal constant, namely, in that it only depends on the type of the quantum phase. 
In fact, for certain models it has been shown that 
\begin{equation}\label{eq:TEEQD}
\gamma=\log\calD,
\end{equation}
where $\calD^2=\sum_{a\in\calL}d_a^2$ is called the total quantum dimension, which is determined by the quantum dimensions $d_a\geq1$ of the anyons emerging in the phase~\cite{KitaevP06,LevinW06} (note that the quantity $\log\calD$ itself is also connected to a secret sharing protocol in another setting~\cite{FiedlerNO17}).   
Hence there are three different ways of obtaining $\gamma$ -- as a universal term in an area law, as a conditional mutual information, and as the logarithm of the total quantum dimension -- that coincide, for example, under the assumption of the area law~\eqref{eq:arealaw}~\cite{BKK19}.
Due to these equivalence relations, not only the subleading term of the area law, but also the conditional mutual information and $\log\calD$ are sometimes called the topological entanglement entropy, depending on the literature. 
Hence one could conjecture that an area law with the subleading term $\gamma$ is indicative of topological order. 

The issue turns out to be more subtle, however.
For example, Bravyi showed that there exists a gapped 2D ground state constructed by a constant-depth local circuit such that the area law has a non-zero constant term for a particular choice of a disc or annulus, while $\log\calD=0$~\cite{BravyiCEX,PhysRevB.94.075151}. 
The constant term (sometimes called ``spurious topological entanglement entropy''~\cite{PhysRevB.94.075151,Williamson18}) also makes the conditional mutual information a nontrivial constant. 
Thus, the subleading term of the area law or the conditional mutual information do not always yield $\log\calD$ for general gapped 2D ground states, and one has to impose additional requirements. 
To the best of our knowledge, these have not been spelled out exactly in the literature.
One issue is that it is often difficult to prove that an area law holds with a universal subleading term.
Hence if one wants to study numerically, it is important to probe the area law for enough distinct regions, to verify that $\gamma$ indeed is universal.
In addition, since the quantity should be topological in nature, it should be invariant under smooth deformations of the boundary.
Bravyi's counterexample does not fulfill this property.

In contrast, we have defined another entropic quantity which has been shown to be an invariant of gapped phases. 
As discussed in Example~\ref{ex:toric1/4}, it takes the same value as the topological entanglement entropy for the toric code. 
Furthermore, we can show that the equivalence also holds for the quantum double model $D(G)$, which is a generalization of the toric code including models with non-abelian anyons: 
\bthm \label{thm:qdouble} For a ground state $|\Omega\>$ of the quantum double model $D(G)$, 
\begin{equation}
\calI(A)_\Omega=\log\calD^2\,
\end{equation}
for any sufficiently large annulus $A$.
\ethm
The proof is in Appendix~\ref{appendix:QD}. 
Unfortunately, as far as we are aware there is no proof yet that the non-abelian quantum double models satisfy the assumption on the stable logical algebra condition, but we believe this to hold. 
Once this has been proven, $\calI(A)_\Omega$ is guaranteed to be $\log\calD$ in any quantum double phase, since it is uniform. 
We emphasize that the uniform property suggests the stable logical algebra condition, since $\calI(A)_\Omega$ provably changes its value if $\calA/\calN_A$ differs depending on the size of the region.

Following these observations, it is natural to expect $\calI(A)_\Omega=\log\calD^2$ holds for more general models. 
Unfortunately, this is not the case; $\calI(A)_\Omega$ is always the logarithm of a rational number for finite $A$, while $\calD^2$ is in general not a rational number, as for example in the case of the double Fibonacci model~\cite{LevinW05}. 
Therefore, we need to generalize the definition to obtain the equivalence for general LCPC models. 
One crucial difference between the quantum double model and the Levin-Wen model is the local degrees of excitations. Excitations are described by ribbon operators in both models, but only those in the (non-abelian) quantum double model have a description of internal degrees of freedom.
More precisely, in non-abelian quantum double models one has to consider ``multiplets'' of independent ribbon operators transforming according to the same charge. 

For models like the double Fibonacci model we might need to consider certain a {\it asymptotic setting}, such as in Refs.~\cite{KatoFM14,KatoFM16,FiedlerNO17}. It could be true that $\calI(A)_\Omega$ is not uniform in these models, and that we recover the relation  $\calI(A)_\Omega=\log\calD^2$ only asymptotically as $A$ grows larger. 
Another possible extension is considering multiple copies of ground states as often considered in quantum Shannon theory. 
Since quantum dimensions represent an asymptotic ratio of the growth of the dimension  of the fusion space, it is reasonable to expect we can obtain an irrational number in a certain asymptotic limit of multiple copies. 
See also example in~\cite{FiedlerNO17} for Fibonacci chain, or Appendix~\ref{app:fibonacci} for a related approach.

It is also natural to expect that there is a quantitative relation between $\calI(A)_\Omega$ and the conditional mutual information. 
Indeed, it has been shown that a similar quantity called the {\it irreducible correlation}~\cite{PhysRevLett.101.180505} equals the conditional mutual information for exactly solvable models (including quantum double models and  Levin-Wen models)~\cite{KatoFM16}. 
The irreducible correlation (of order 3) $C^{(3)}(\rho_{XYZ})$ of a tripartite state $\rho_{XYZ}$ is defined as 
\begin{equation}
C^{(3)}(\rho_{XYZ}):=S(\rho_{XYZ}\|{\widetilde\rho}_{XYZ})\,,
\end{equation}
where ${\widetilde \rho}_{XYZ}$ is the maximum entropy state defined by
\begin{equation}
{\widetilde \rho}_{XYZ}:=\operatornamewithlimits{argmax}_{\sigma_{XYZ}\in R_2}\, S(\sigma_{XYZ}) 
\end{equation}
with $R_2:=\{\sigma_{XYZ}|\sigma_R=\rho_R \,,\; R=XY,YZ,ZX\}$. 
Hence $R_2$ is the sets of all tripartite states that agree with $\rho_{XYZ}$ when tracing out one of the parts.
This set is convex and the maximum entropy state is unique. 
\begin{prop}~\cite{KatoFM16} For a ground state satisfying an exact area law: $S(\rho_A)=\alpha|\partial A|-n_A\gamma$, it holds that 
\begin{equation}
C^{(3)}(\rho_{XYZ})=I(X:Z|Y)_\Omega
\end{equation}
for any tripartition of an annular region such that $Y$ separates $X$ from $Z$. 
\end{prop}
Bravyi's counter example satisfies the condition of this theorem. Therefore, $C^{(3)}(\rho_{XYZ})$ is not an invariant of gapped phases and neither is $I(X:Z|Y)_\Omega$. More precisely, Bravyi gives an example of a state where these quantities are non-zero, but which nevertheless is not topologically ordered. 
The information convex $\Sigma(A)$ is a subset of $R_2$ under an appropriate partition of $A$. Indeed, it is a strict subset in the case of Bravyi's counter example, in which $C^{(3)}(\rho_{XYZ})=I(X:Y|Z)_\Omega$ takes a non-trivial value while $\calI(A)_\Omega=0$.  
Therefore these two equivalent quantities quantify not only ``topological'' contributions but can also contain ``non-topological'' contributions. We expect that $\calI(A)_\Omega$ (or a suitable generalization of it) captures the topological part and it provides a lower bound of these quantities:  
\begin{conj}\label{conj} Suppose $|\Omega\>$ is a ground state of a Hamiltonian satisfying all assumptions. 
For any tripartition $XYZ$ of $A$ such that $Y$ separates $X$ from $Z$ as depicted in Fig.~\ref{fig:TEE1}, it holds that
\begin{equation}
I(X:Z|Y)_\Omega\geq \calI(A)_\Omega\,.
\end{equation}
\end{conj}

A similar bound is easy to check for the entropy difference instead of the relative entropy:
\begin{prop}
We have the following lower bound 
\begin{equation}
I(X:Z|Y)_\Omega\geq S(\tau_A)-S(\rho_A)\,,
\end{equation}
for for any tripartition $XYZ$ of $A$ such that $Y$ separates $X$ from $Z$ as depicted in Fig.~\ref{fig:TEE1}.
\end{prop}
\begin{proof} 
By LTQO', any $|\phi\>\in\calH_\Omega$ is indistinguishable from $|\Omega\>$ for any subregion of $XYZ$, e.g., $XY$ which has trivial topology. 
Hence, $\tau_A$ has the same local marginals as of $\rho_A$. From the strong subadditivity for system $A=XYZ$, we have
\begin{align}
S(\tau_A)-S(\rho_A)&\leq S(\tau_{XY})+S(\tau_{YZ})-S(\tau_Y)-S(\rho_A)\\
&=S(\rho_{XY})+S(\rho_{YZ})-S(\rho_Y)-S(\rho_{XYZ})\\
&=I(X:Z|Y)_\Omega\,.
\end{align}
Note that we can use LTQO' for e.g., $Y$ in Fig.~\ref{fig:TEE1} which has multiple connected components, since the reduced state is a product of that of connected regions.
\end{proof}
Therefore, the conjecture holds if the relative entropy difference is equal to the entropy difference, i.e.,
\begin{equation}
S(\tau_A)-S(\rho_A)=S(\rho_A\|\tau_A)
\end{equation}
holds. This condition is known to be satisfied when $\tau_A$ is the maximum entropy state of a convex set containing $\rho_A$ which is defined by linear constraints~\cite{Weis10}. 
In our case, the convex set is the information convex $\Sigma(A)$. 
Note that the entropy difference between the maximum entropy state in $\Sigma(A)$ and $\rho_A$ has been studied in Ref.~\cite{Shi18}. 
While $\tau_A$ coincides with the maximum entropy state in quantum double models, it is still unclear that the equivalence is stable under constant-depth local circuits.

\section{Discussion}
In this paper, we have introduced an entropic quantity $\calI(A)_\Omega$ of 2D gapped phase described by LCPC Hamiltonians. $\calI(A)_\Omega$ is defined based on the operator algebras of logical operators defined for annulus, and it is invariant under constant-depth local circuits.  
We have also shown that $\calI(A)_\Omega$ is equivalent to the logarithm of the ratio of the corresponding superselection sectors defined via the GNS Hilbert space constructed by the algebras.   
We have demonstrated that $\calI(A)_\Omega$ matches $\log\calD^2$ for the toric code model, or more generally the quantum double models $D(G)$, including models with non-abelian anyons. 

Several questions still remain. Especially, it is desirable to extend the framework so that the equality with $\log\calD^2$ holds for models with irrational quantum dimension, like the double Fibonacci model. To obtain an irrational number, we might need to consider the superselection sectors in some asymptotic setting, since it represents the asymptotic growth ratio of the dimension of certain Hilbert space in the modular tensor category description. Another important direction is proving Conjecture~\ref{conj}. Again, it might be true only in a certain asymptotic scenario. Once the conjecture will be shown, we have a decomposition of the conditional mutual information into ``topological'' contribution and ``non-topological'' contribution. While the known fixed-point models of non-chiral topologically ordered phases are described by LCPC Hamiltonian, it would be desirable to extend our framework to general frustration-free Hamiltonians.  Indeed, an extension of the information convex for frustration-free Hamiltonians is discussed in Ref.~\cite{Shi18}. However, it is unclear if the corresponding logical operators form a proper $C^*$-algebra. 

Superselection sectors for anyon models are also considered in the thermodynamic limit (for infinitely large spin systems or in algebraic quantum field theory). 
In these theories,  factors of von Neumann algebras, which are subalgebras containing the identity as the center, play an important role. 
In contrast, we are considering finite-dimensional algebras and we do not have factors. 
As discussed in Appendix~\ref{sec:thermodynamic}, the relative entropy, the total quantum dimension and the so-called Jones index are mutually connected (see also Eq.~(7) of Ref.~\cite{FiedlerNO17}). 
Connecting our theory of finite-dimensional framework to these infinite-dimensional framework is desirable to obtain the most general understanding of the origin of the topological entanglement entropy.

\begin{acknowledgments}
KK is thankful to Bowen Shi for helpful discussions and sharing his notes on calculation of the information convex. KK acknowledges funding provided by the Institute for Quantum Information and Matter, an NSF Physics Frontiers Center (NSF Grant PHY-1733907) and JSPS KAKENHI Grant Number JP16J05374.
PN has received funding from the European Union’s Horizon 2020 research and innovation program under the Marie Sklodowska-Curie grant agreement No 657004 and the European Research Council (ERC) Consolidator Grant GAPS (No. 648913). 
\end{acknowledgments}

\appendix

\section{Logical algebra and the GNS construction}\label{app:opalg}
In this appendix we discuss some of the more technical properties of the logical algebras.
In particular, we show that for the \emph{stable logical algebra condition} it is enough to check if the inclusion map induces an isomorphism of the logical algebras related to an inclusion of cones.
We also outline the basics of the \emph{GNS construction}, which gives a canonical way to obtain a Hilbert space and an representation given a state on an abstract $C^*$-algebra.
In the applications that we have in mind this gives rise to the appearance of different superselection sectors.

\subsection{Stable logical algebra condition}
In Section~\ref{sec:logicalalgebra} we introduced \emph{stable logical algebra condition}.
This says that the logical algebras for different (large enough) annuli are isomorphic.
This is true in particular if we have an inclusion $A_1 \subset A_2$ of annuli.
This induces in inclusion of the corresponding algebras of observables.
What is not so clear, however, is that this embedding in fact induces an isomorphism of the corresponding logical algebras.
The following proposition shows that this is nevertheless the case, hence it is enough to check the stable logical algebra condition for this particular map.\\

\noindent{\bf Proposition~\ref{prop}} {\it
Let $A_1 \subset A_2$ be two annuli and assume the uniform stable logical algebra condition~\eqref{def:stability}.
Then the identity map provides a natural embedding $\iota : \calA_1 \to \calA_2$, which induces an isomorphism $\calA_1/\calN_{A_1} \to \calA_2/\calN_{A_2}$ of the quotient algebras, where we used the notation of Section~\ref{sec:logicalalgebra}.
}
\begin{proof} 
Note that $\calA, \calE$ and $\calN$ all depend on the choice of annulus.
We will write $\calA_i, \calE_i$ and $\calN_i$ for the corresponding algebras and ideals, and define $\calN_{A_i} \equiv \calN_i \cap \calA_i$.
The corresponding equivalence classes  $\calA_i/\calN_{A_i}$ are written $[A]_i$.

First note that if $A \in \calA_1$ it follows that $[A, h_j] = 0$ for all $h_j$ with $\supp(h_j) \subset A_{2,+}$, either by locality or because $A \in \calA_1$.
Hence there is a linear map $\iota_{A_1 A_2} : \calA_1 \to \calA_2$ given by the natural inclusion.
We claim that this induces an inclusion of the quotient algebras $\calA_1/\calN_{A_1} \to \calA_2/\calN_{A_2}$.
Note that it is enough to show that $\calN_{A_1} \subset \calN_{A_2}$.
If $N \in \calN_{A_1}$, then by definition $\supp(N) \subset A_1$ and $N \Pi_{A_{1,+}} = 0$.
But by locality $[N, h_j] = 0$ if $\supp(h_j) \subset A_1^c$.
Because by assumption all $h_j$ commute and the model is frustration free, we have
\begin{equation}
	\label{eq:prodprojection}
	N \Pi_{A_{2,+}} = N \left( \prod_{\supp(h_j) \subset (A_2 \setminus A_{1,+})_+} h_j \right) \Pi_{A_{1,+}} = \left( \prod_{\supp(h_j) \subset (A_2 \setminus A_{1,+})_+} h_j \right) N \Pi_{A_{1,+}} = 0,
\end{equation}
and hence $N \in \calN_{A_2}$.

We now claim that this map on the quotient algebras is injective.
To this end, let $A,B \in \calA_1$ and suppose that $[A]_2 = [B]_2$.
Then $A-B = N$ for some $N \in \calN_2$.
Note that the left-hand side is supported on the annulus $A_1$, and hence by locality, commutes with any operator supported on the complement of $A_1$.
But this implies that $\supp(N) \subset A_1$, and it remains to be shown that $N \Pi_{A_{1,+}} = 0$.

To reach a contradiction, suppose that $N \Pi_{A_{1,+}} \neq 0$.
As in equation~\eqref{eq:prodprojection}, we can write $N \Pi_{A_{2,+}}$ as a product of $N \Pi_{A_{1,+}}$ and a projection which we will write as $\Pi_{(A_2\setminus A_{1,+})_+}$.
Because $\supp(N) \subset A_1$ and the terms $h_j$ in the Hamiltonian mutually commute, it follows that $N \Pi_{A_{1,+}}$ and $\Pi_{(A_2\setminus A_{1,+})_+}$.
Note however that because there may be terms $h_j$ outside of $A_1$ with $\supp(h_j) \cap A_{1,+} \neq \emptyset$, hence the two operators do not have disjoint support.
Nevertheless, one can actually factor the Hilbert space such that these operators act on different tensor factors, by the commuting property and because our algebras are finite dimensional.
Indeed, this follows from Theorem~1 of~\cite{ScholzWerner08}.
From this tensor product decomposition we see that $N\Pi_{A_{2,+}}\neq0$ if $N\Pi_{A_{1,+}}\neq0$, and hence we conclude by contradiction that $N \Pi_{A_{1,+}} = 0$ and $N \in \calN_{A_1}$.

This leads to the conclusion that the inclusion map $\iota$ induces an inclusion on the quotient algebras.
Because the algebras $\calA_i/\calN_{A_{i,+}}$ are finite dimensional and isomorphic by assumption, this map must be surjective as well.
This completes the proof.
\end{proof}

\subsection{Relation to GNS construction}
In Section~\ref{sec:setting} we defined the Hilbert space $\mathcal{H}_\Omega$, capturing the different superselection sectors.
This construction can be understood naturally in terms of representations of operator algebras.
The essential idea behind the construction is that the algebra of observables which do not change the total charge inside the annulus, which we called $\mathcal{C}$ before, is represented on different ways on the Hilbert space $\mathcal{H}_\Omega$, corresponding to the different sectors.
Here give a brief introduction to the main aspects needed to understand this connection.

In an operator algebraic approach it is often convenient to consider abstract algebras without any reference to a Hilbert space.
A state in that setting is then a positive linear functional $\omega$ on the algebra, normalized such that $\omega(I) = 1$.
For finite dimensional algebras this is equivalent to $\omega(A) = \Tr(\rho A)$ for some density matrix $\rho$ with unit trace, but for infinite systems not all states are of this form.
On the other hand, the Hilbert space picture is very useful in quantum mechanics.
Hence it is useful to go from the abstract picture back to the Hilbert space picture in a canonical way.
The \emph{Gel'fand-Naimark-Segal} (GNS) construction provides such a method.
It yields a representation of a $C^*$-algebra $\calM$ as bounded operators on some Hilbert space, in such a way that $\omega$ is represented by a vector in this Hilbert space.
In a sense it can be understood as a form of \emph{purification}, although if the representation is not irreducible, the state is not pure (seen as a state of $\calM$).
It is a standard tool in operator algebra, and can be found in most textbooks on the subject, see for example~\cite{TakesakiI,Naaijkens17}.

For simplicity we only consider the case that $\mathcal{M}$ is a finite dimensional, unital algebra.
It follows that $\mathcal{M}$ is of the form $\bigoplus_k M_{n_k}(\mathbb{C})$ for a (unique, up to permutation) finite sequence of integers $n_k$.
Let $\omega$ be a state on $\mathcal{M}$, that is $\omega(A) = \Tr(\rho A)$ for some positive operator $\rho \in \mathcal{M}$ of unit trace, and all $A \in \mathcal{M}$.
The goal is to define a new Hilbert space $\mathcal{H}_\omega$ and a representation $\pi_\omega : \mathcal{M} \to \mathcal{B}(\mathcal{H}_\omega)$ such that there is a vector $|\Omega\> \in \mathcal{H}_\omega$ with $\omega(A) = \< \Omega | \pi_\omega(A) | \Omega \>$ for all $A \in \calM$.
In other words, in the new representation the state is represented by a vector.
Moreover, this vector is \emph{cyclic}, in the sense that $\pi_\omega(\calA)|\Omega\> = \mathcal{H}_\omega$.

To define the Hilbert space, first we define
\begin{equation}
	\calJ := \{ A \in \calM : \omega(A^\dagger A) = 0 \}.
\end{equation}
It follows from the Cauchy-Schwarz inequality of positive linear functionals that $\calJ$ is a linear space.
In fact, it can be shown that $\calJ$ is a left ideal of $\calM$, in the sense that $AJ \in \calJ$ for all $A \in \calM$ and $J \in \calJ$.
The Hilbert space $\mathcal{H}_\omega$ is then defined to be the quotient (as a vector space) $\mathcal{H}_\omega := \cal{M}/\mathcal{J}$.
We will write $|[A]\>$ for the equivalence class of a representative $A \in \calM$.
The definition of $\mathcal{H}_\omega$ is complete by defining an inner product, by setting $\< [A] | [B] \> := \omega(A^\dagger B)$.
By the remark above this is well-defined.
The inner product is also non-degenerate precisely because we divide out the ideal $\calJ$, which correspond to vectors of length zero.

The representation of $\calM$ can be defined by its action on the vectors of $\calH$: $\pi_\omega(A) |[B]\> := |[AB]\>$.
Again, this is well-defined because $\calJ$ is a left ideal.
It is also straightforward to check that $\pi_\omega$ is linear, $\pi_\omega(AB) = \pi_\omega(A) \pi_\omega(B)$ and $\pi_\omega(A^\dagger) = \pi_\omega(A)^\dagger$.
Hence $\pi_\omega$ is a representation.
Finally, $|\Omega\> := |[I]\>$ has the properties claimed.

Before we discuss an example, we first mention two more properties of the GNS construction.
Firstly, it is unique up to unitary equivalence: if $(\pi_\omega', \mathcal{H}_\omega', \Omega')$ is another triple, then there is a unitary $U: \calH_\omega' \to \calH_\omega$ such that $\pi_\omega(A) = U \pi_\omega'(A) U^\dagger$ and $U |\Omega'\> = |\Omega\>$.
Secondly, $\pi_\omega$ is an irreducible representation if and only if $\omega$ is a pure state, in the sense that it cannot be written as a convex combination of two distinct states.
In the present setting, this is equivalent to $\pi_\omega(\calA) \cong M_k(\mathbb{C})$ for some integer $k$.

We now apply this construction to the algebras defined in Section~\ref{sec:setting}.
In particular, let $|\Omega\rangle$ be a ground state of the Hamiltonian.
This induces a state on $\calE$ by $\omega(O):=\<\Omega|O|\Omega\>$. 
We denote the corresponding GNS triplet $(\pi^\Omega, \calH_\Omega, |[I]\>)$. 
The GNS Hilbert space has a physical interpretation established by the following isomorphism.

\blm \label{lemma:GNS=spanE}
There is an isomorphism of Hilbert spaces such that 
\begin{equation}\label{eq:eqGNS1}
\calH_\Omega\cong \calE|\Omega\>=\left\{|\psi\>\in\calH\,\left|\;  h_j|\psi\>=|\psi\> \quad {\rm if }\;\supp(h_j)\subset A_+ \right.\right\}\,.
\end{equation}
Moreover, $\pi^\Omega$ is an irreducible representation of $\calE$.
\elm
Note: this essentially follows from the uniqueness of the GNS representation (up to unitary equivalence), but we give an explicit proof for the benefit of the reader.
\begin{proof}
 By construction, $\calH_\Omega:=\calE/\calJ=\pi^\Omega(\calE)|[I]\>$, where $\calJ$ is the left ideal as defined above.  
We first show that the linear map $\iota:\calH_\Omega\to \calE|\Omega\>$ defined by $\iota:|[O]\>\mapsto O|\Omega\>$ 
is an isomorphism. 
The map is well-defined: if $[O_1] = [O_2]$, then $O_1 = O_2 + J$, with $J \in \calJ$.
But $\omega(J^\dagger J) = \< J \Omega | J \Omega \> = 0$, and hence $J |\Omega\> = 0$.
Since $\iota$ is defined for all $O\in \calE$, it is surjective. To see the map is also injective, it is enough to check that the equivalence 
\begin{equation}
O_1|\Omega\>=O_2|\Omega\> \Leftrightarrow (O_1-O_2)|\Omega\>=0
\end{equation}
implies that $[O_1] = [O_2]$.
From $(O_1-O_2)|\Omega\> = 0$ it follows that $\omega((O_1-O_2)^\dagger(O_1-O_2)) = 0$, and hence $O_1-O_2 \in \calJ$. 
Finally, $\<[O_1]|[O_2]\> = \omega(O_1^\dagger O_2) = \<\Omega|O_1^\dagger O_2 |\Omega\>$, hence the inner product is preserved by $\iota$.

The second equality holds because we have $|\psi\>\<\Omega|\in \calE$ for  any $|\psi\>$ such that $H_{A_+}|\psi\>=|\psi\>$, and thus $(|\psi\>\<\Omega|)|\Omega\>=|\psi\>$. This implies 
\begin{equation}
\left\{|\psi\>\in\calH\,\left|\;  h_j|\psi\>=|\psi\> \quad {\rm if }\;\supp(h_j)\subset A_+ \right.\right\}\subset\calE|\Omega\>
\end{equation} 
(the other inclusion is easy to check).  Operators like $|\psi\>\<\Omega|$ and their conjugates span the full-matrix algebra on $\calE|\Omega\>$.
We also note that $\iota( \pi^\Omega(O_1)|[O_2]\> ) = \iota( |[O_1 O_2] \> ) = O_1 O_2 |\Omega\> = O_1 \iota(|[O_2]\>)$. Hence $\iota$ is compatible with the representation $\pi^\Omega$.
This together with the equivalence relation~\eqref{eq:eqGNS1} show that any vector in $\calE|\Omega\>$ is cyclic for $\pi^\Omega$, and it follows that is an irreducible representation of $\calE$ on $\calH_\Omega$. 
\end{proof}

We conclude with an observation on $\pi^\Omega$ which turns out to be useful.
\begin{lem}\label{lem:nkernel}
The representation $\pi^\Omega$ can be restricted to $\calE/\calN$.
\end{lem}
\begin{proof}
We show that $\calN$ is in the kernel of $\pi^\Omega$.
Indeed, by definition of the GNS triplet, we have
\begin{equation}
A\in \Ker \pi^\Omega \Leftrightarrow \pi^\Omega(A)|[O]\>=|[AO]\>=0, \;\forall O\in \calE\,.
\end{equation}
From the equivalence~\eqref{eq:eqGNS1}, the second condition is equivalent to
\begin{equation}
AO|\Omega\>=AO\Pi_{A_+}|\Omega\>=A\Pi_{A_+}O|\Omega\>=0,\; \forall O\in \calE\,.
\end{equation}
Therefore, if $A\in\calN$, then $A\in \Ker \pi^{\Omega}$. This inclusion implies $\pi^\Omega([A]_\calN):=\pi^\Omega(A)$ for $[A]_\calN\in \calE/\calN$ is a well-defined representation of $\calE/\calN$. 
\end{proof}
This is in fact the GNS representation of $\calE/\calN$ obtained by regarding $\omega([A]):=\omega(A)$ as a state on $\calE/\calN$. 

\section{Calculation of the Invariant for Quantum Double Models}\label{appendix:QD}
In this appendix, we will show that 
\begin{equation}
\calI_\Omega(A)=2\log\calD\,
\end{equation}
for $D(G)$, in which we have $\calD=|G|$. We will start from the definition of $D(G)$, and then reveal all extremal points of the information convex in the next subsection. 
We calculate $\calI(A)_\Omega$ in the last subsection. 

\subsection{Quantum Double Model \texorpdfstring{$D(G)$}{D(G)}}
The quantum double model $D(G)$, defined for any finite group $G$, is a generalization of the toric code model (which corresponds to $G=\mathbb{Z}_2$)~\cite{Kitaev03}. 
We recall the main definitions here.
The model is defined on a directed graph, where a Hilbert space ${\mathbb C}[G]=\Span\{|g\>|g\in G\}$ is associated to each edge as in Fig.~\ref{fig:QDlat}. 
For simplicity we assume a square graph with the same orientation as in the figure, but the model can be defined for more general directed graphs.
The left (right) multiplication operator is denoted by $L_g^+:=\sum_{h\in G}|gh\>\<h|$ $(L^-_g:=\sum_{h\in G}|hg^{-1}\>\<h|)$. 
Also, we denote projectors on to a group element by $T_g^+:=|g\>\<g|$ $(T_g^-:=|g^{-1}\>\<g^{-1}|)$. 
In a similar way to as in the toric code, the Hamiltonian is defined by vertex operators $A_v$ and plaquette operators $B_p$, which are defined as
\begin{align}
A_v:=\frac{1}{|G|}\sum_{g\in G}A_v^g, \quad B_p:=\sum_{h_1h_2h_3h_4=e}T_{h_1}^+T_{h_2}^+T_{h_3}^-T_{h_4}^-\,,
\end{align}
where $A_v^g=L_g^-L_g^+L_g^+L_g^-$, where the first operator $L_g^-$ acts on the site at the left side of $v$ and the rest in clockwise order (Fig.~\ref{fig:QDlat}). 
The $\pm$ signs of the multiplication operators are determined by whether the site is on an incoming edge ($+$) or on an outgoing edge ($-$). 
The direction of the edges can be changed by the corresponding local unitary, which maps $|g\>\mapsto |g^{-1}\>$.

Excitations of $D(G)$ are labeled by the irreducible representations of a Hopf algebra called the quantum double, first introduced by Drinfel'd~\cite{Drinfeld87}. They are specified by pairs $(R,C)$, where $C$ is a conjugacy class of $G$ and $R$ is an irreducible representation of the centralizer group of $C$~\cite{DijkgraafPR91}. We denote the set of all conjugacy class of $G$ by $(G)_{cj}$. For each conjugacy class $C\in (G)_{cj}$, we fix a representative $r_C$ and denote its centralizer group by $E(C):=\{g\,|\,gr_Cg^{-1}=r_C\}$. The elements of the conjugacy class are in one-to-one correspondence with the cosets $G/E(C)$. The quantum dimension of the charge $(R,C)$ is given by $d_{(R,C)}=n_R|C|$, where $n_R$ is the dimension of $R$. From the representation theory of finite groups, we always have  $\sum_{(R,C)}n^2_R|C|^2=|G|^2$. 

While excitations in the toric code model $D(\mathbb Z_2)$ are created by string operators or dual string operators, excitations in general $D(G)$ models are created by ribbon operators $\{\calF^{(h,g)}_\rho\}$, where $h,g\in G$ and $\rho$ is a ``ribbon'', a combination of a neighboring string and dual string. See e.g., Refs.~\cite{Kitaev03,BombinMD08} for more details.

\begin{figure}[htbp] 
\begin{center}
\includegraphics[width=6.0cm]{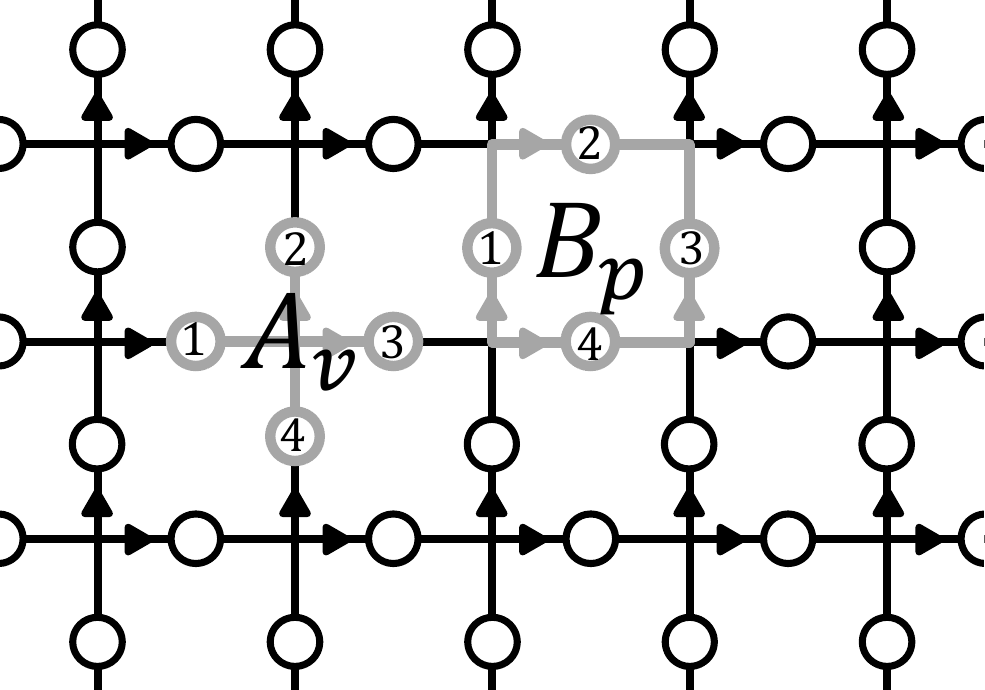}
\end{center}
\vspace{-4mm}
\caption{The quantum double model defined on a directed square lattice. Each $A_v$ acts on 4 sites around vertex $v$ and $B_p$ acts on 4 sites around plaquette $p$.  }
\label{fig:QDlat}
\end{figure}

\subsection{Calculation of \texorpdfstring{$\Sigma(A)$}{Sigma(A)}}
As already mentioned, the structure of information convex $\Sigma(A)$ for $D(G)$ has been derived in Ref.~\cite{Shi18}. In this section we explicitly calculate $\Sigma(A)$ for concreteness. Some of the techniques in the calculation will be used in the calculation of $\calI(A)_\Omega$. 
For simplicity we only consider the thinnest rectangular annulus $A$ as in Fig.~\ref{fig:QDthin}, but a similar argument can be applied to a general annulus~\cite{Shi18,BShipriv}. See the end of Appendix~\ref{appendix:QD}  for more details. 

Let us consider the quantum double model $D(G)$ defined on a square lattice embedded on a sphere. As shown in Lemma~\ref{lemma:GNS=spanE}, the GNS Hilbert space $\calH_\Omega$ for a ground state $|\Omega\>$ is equivalent to the ground subspace of  $H_{A_+}$:
\begin{equation}
\calH_\Omega\cong\left\{|\phi\>\in\calH \left| \Pi_{A_+}|\phi\>=|\phi\> \right.\right\}.
\end{equation}

We label the basis elements of sites at the inner boundary by $h_1,h_2,...,h_{n+4}$ and sites at the outer boundary by $H_1,...,H_N$, where $h_i,H_j\in G$, in such a way that the direction at the boundaries are aligned as depicted in Fig.~\ref{fig:QDthin}.  
We especially choose one site in the bulk of $A$ and label it by $t$. Other sites in $A$ are labeled by $g_1,...,g_m$. In this notation, a basis of $\calH_{A}$ is written as 
\begin{equation}\label{eq:labelA}
|\{h_i\},\{H_j\},\{g_k\},t\>_A\,,
\end{equation}
where $\{h_i\}=\{h_1,\dots,h_{n+4}\}$ and so on. The annulus $A$ contains $n+2N+4$ spins, the support of $N-4$ plaquette operators and $4$ vertex operators at the inner corners. 

When we restrict to the ground subspace of all $B_p$ such that $\supp(B_p)\subset A$, every $g_k$ is uniquely determined by $\{h_i\}, \{H_j\}$ and $t$.  For instance, $g_1=h_1tH_1^{-1}$ and $g_2=h_2h_1tH_1^{-1}H_2^{-1}$. For this reason, we will omit $\{g_i\}$ from the notation Eq.~\eqref{eq:labelA} in the following. The products $h:=h_1h_2 \cdots h_{n+4}$ and $H:=H_1H_2\cdots H_N$ are also restricted by the product of $B_p$ so that $h=tHt^{-1}$, which implies that $h$ and $H$ are in the same conjugacy class. For $C\in(G)_{cj}$, there exists a set of group elements $\{q_i\}_{i=1}^{|C|}$ such that   $h=q_ir_Cq_i^{-1}$, $H=q_jr_Cq_j^{-1}$ and $t=q_i{\bar t}q^{-1}_j$ for every $h,H\in C$ and  ${\bar t}\in E(C)$.

\begin{figure}[htbp] 
\begin{center}
\includegraphics[width=6.0cm]{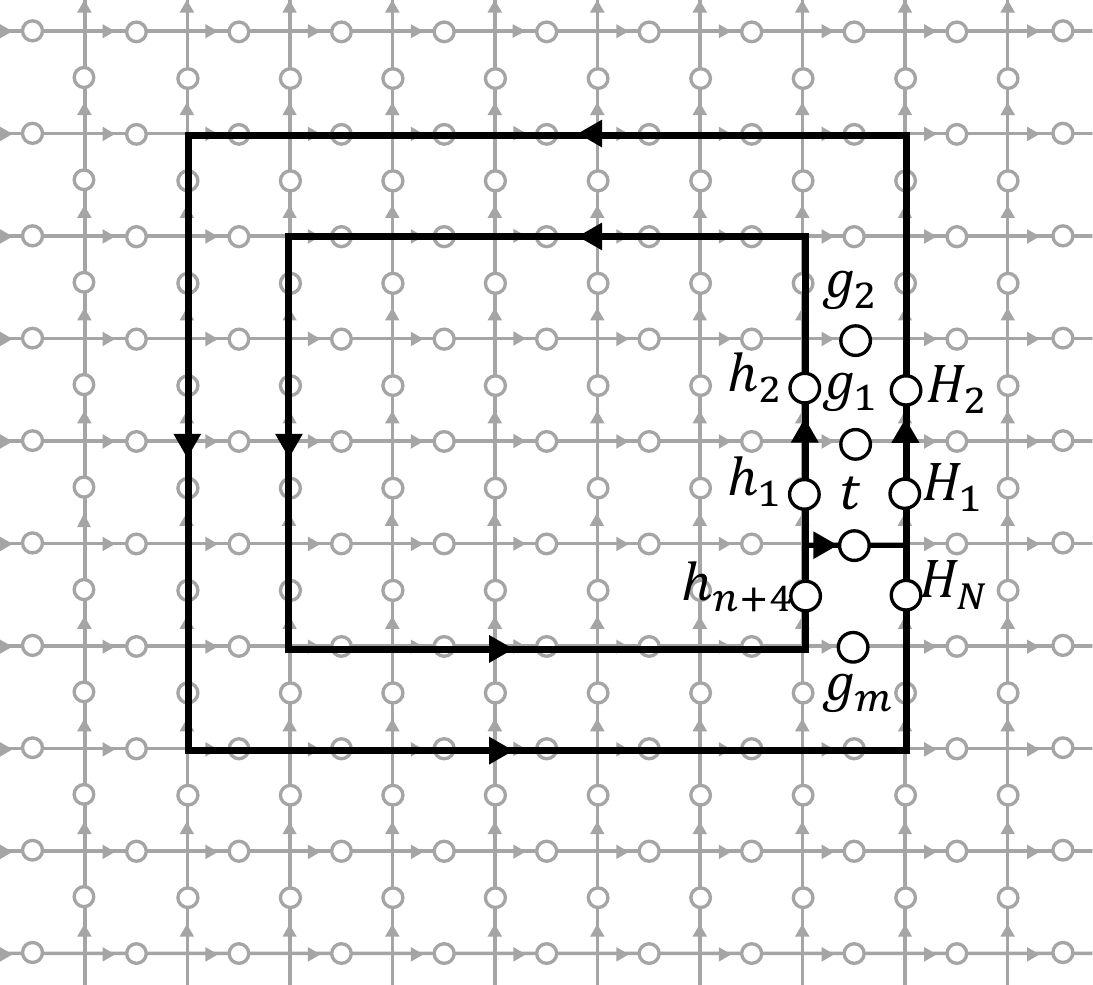}
\end{center}
\vspace{-4mm}
\caption{An example of the thinnest rectangular region $A$ (surrounded by thick lines). We choose the directions of both boundaries the same. We label the sites at the inner boundary by $h_i$ and those at the outer boundary by $H_j$. We choose one site inside the region and label it by $t$.}
\label{fig:QDthin}
\end{figure}

We then consider four vertex operators at the inner corners of $A$. 
Let us denote the vertices along the inner boundary by $v_1,v_2,\dots,v_{n+4}$ with counter clockwise order from above $h_1$. 
Suppose $h_i$ and $h_{i+1}$ denotes the inner boundary sites at the upper right corner. When we apply $A_{v_i}^g$, they are mapped to $h_ig^{-1}$ and $gh_{i+1}$, and therefore the product $h_ih_{i+1}$ is preserved under the action. 
Indeed, by properly choosing $g\in G$, $A_{v_i}^g$ can maps $\{h_i,h_{i+1}\}$ to any other pair $\{h'_i,h'_{i+1}\}$ satisfying $h_ih_{i+1}=h'_{i}h'_{i+1}$. 
Therefore, a linear combination of vectors $|\{h_i\},\{H_j\},t\>_A$ is stabilized by $A_{v_i}$ if and only if all two configurations $\{h_i\},\{h'_i\}$ satisfying $h_{i}h_{i+1}=h'_{i}h'_{i+1}$ appear in an equal weight.  We introduce such a state by
\begin{equation}
|\{h_i\},\{H_j\},t\>_A\equiv \frac{1}{\sqrt{|G|}}\sum_{{\tilde h}_i{\tilde h}_{i+1}=h_i}|\{{\tilde h}_i\},\{H_j\},t\>_A\,,
\end{equation}
where ${\tilde h}_k=h_k$ for $k<i$, ${\tilde h}_k=h_{k-1}$ for $k>i+1$ and $|\{h_i\}|=n+3$.
By definition, $A_{v_i}|\{h_i\},\{H_j\},t\>=|\{h_i\},\{H_j\},t\>$ for every $\{h_i\}$. We repeat the same procedure for all other corners and define a new label set $\{h_i\}=\{h_1,h_2,\dots,h_n\}$ which has $n$ independent elements. Hence $\{|\{h_i\},\{H_j\},t\>_A\}$ spans $(n+N)$-dimensional space.

Next, we explicitly calculate possible states in the information convex $\Sigma(A)$~\eqref{def:infoconv}. Without loss of generality, we only care about $A_+$ and denote $B:=A_+\backslash A$. This is because we have $\calH_\Omega\cong\calK_{AB}\ot{\calH}_{(A_+)^c}$, where 
\begin{equation}
\calK_{AB}:=\left\{\left.|\phi\>_{A_+}\in\calH_{AB}\,\right|\, H_{A_+}|\phi\>_{A_+}=|\phi\>_{A_+}\right\}\,.
\end{equation}
We decompose vertex operators supported on both $A$ and $B$ as $A_v^g={\tilde A}_v^g\ot {\bar A}^g_v$, so that ${\tilde A}_v^g$ $({\bar A}_v^g)$ only acts on $A$ ($B$). 
 If $|\psi\>\in\calH_A\ot\calH_B$ satisfies $A_v|\psi\>=|\psi\>$, it holds $A_v^g|\psi\>=A_v^gA_v|\psi\>=A_v|\psi\>=|\psi\>$.  From the unitarity of $A_v^g$, we have
\begin{equation}
A_v|\psi\>=|\psi\>\Rightarrow {\tilde A}_v^g\psi_A{\tilde A}_v^{g^{-1}}=\psi_A \,,\;\forall g\in G.
\end{equation}
Here we used $(L_\pm^g)^\dagger=L_\pm^{g^{-1}}$. 
Therefore, states in $\Sigma(A)$ should be invariant under all operators in the form: ${\tilde A}_{v_1}^{l_1}\ot {\tilde A}_{v_2}^{l_2}\ot\cdots {\tilde A}_{v_{n-1}}^{l_{n-1}}$. 
$\{h_i\}$ is mapped by these unitaries to $(h_1l_1^{-1}),(l_1h_2l_2^{-1}),...,(l_{n-1}h_n)$ and thus $h=h_1...h_n$ is unchanged. There are $|G|^{n-1}$ patterns of the choice of $\{h_i\}$ for fixed $h$, and every two patterns are mapped each other by these products of ${\tilde A}_v^g$. The same argument holds for the outer boundary with vertices $V_1,V_2,...,V_N$. Hence, if $h=tHt^{-1}$, the equal weight mixture
\begin{equation}\label{eq:eqspohHt}
\sum_{h_1\cdots h_n=h}\sum_{H_1 \cdots H_N=H}|\{h_i\},\{H_j\},t\>\<\{h_i\},\{H_j\},t|_A
\end{equation}
is invariant under all vertex operators except on $v_n$ and $V_N$.

To see the action of the remaining vertex operators clearer, we introduce the group Fourier basis:
\begin{equation}
|R;a,b\>:=\sqrt{\frac{n_R}{|G|}} \sum_{g\in G}R_{ab}(g)|g\>\,,
\end{equation}
where $R_{ab}(g)$ is the $(a,b)$ matrix element of the irreducible representation $R$ of $G$ with dimension $n_R$. 
We define a new basis by
\begin{equation}\label{eq:RCbasis}
\left|(R,C);u,v;\{h_k\}_{q_i},\{H_l\}_{q_j}\right\>_A:=\sum_{t\in E(C)}\sqrt{\frac{n_R}{|E(C)|}}R_{ab}(t)|\{h_k\}_{q_i},\{H_l\}_{q_j},q_itq_j^{-1}\>_A\,
\end{equation}
for $C\in(G)_{cj}$, $R\in(E(C))_{ir}$, $u=(q_i,q_j)$ and $v=(a,b)$. Here, $\{h_k\}_{q_i}$ denotes a set $\{h_1,...,h_n\}$ satisfying $h_1...h_n=q_ir_Cq_i^{-1}$ and $\{H_l\}_{q_j}=\{H_1,...,H_N\}$ with $H_1...H_N=q_jr_Cq_j^{-1}$. 
${\tilde A}_{v_n}^{g}\ot {\tilde A}_{V_N}^{g'}$ maps $h$, $H$ and $t$ to $ghg^{-1}$, $g'Hg'^{-1}$ and $gtg'^{-1}$, respectively. 
In other words, there are $t_1,t_2\in E(C)$ such that $gq_i=q_{i'}t_1$ and $g'q_j=q_{j'}t_2$, and 
\begin{align}
{\tilde A}_{v_n}^{g}\ot {\tilde A}_{V_N}^{g'}\left|(R,C);u,v;\{h_k\}_{q_i},\{H_l\}_{q_j}\right\>_A&=\sum_{t\in E(C)}\sqrt{\frac{n_R}{|E(C)|}}R_{ab}(t)|\{h_k\}_{q_{i'}},\{H_l\}_{q_{j'}},q_{i'}t_1tt_2^{-1}q_{j'}^{-1}\>_A\\
&=\sum_{t\in E(C)}\sqrt{\frac{n_R}{|E(C)|}}R_{ab}(t_1^{-1}tt_2)|\{h_k\}_{q_{i'}},\{H_l\}_{q_{j'}},q_{i'}tq_{j'}^{-1}\>_A\\
&=\sum_{c,d}R_{ac}(t_1^{-1})R_{db}(t_2)\sum_{t\in E(C)}\sqrt{\frac{n_R}{|E(C)|}}R_{cd}(t)|\{h_k\}_{q_{i'}},\{H_l\}_{q_{j'}},q_{i'}tq_{j'}^{-1}\>_A\\
&=\sum_{v'}U_{vv'}\left|(R,C);u',v';\{h_k\}_{q_{i'}},\{H_l\}_{q_{j'}}\right\>\,,%
\end{align}
where $u'=(q_{i'},q_{j'})$, $v'=(c,d)$ and $U_{vv'}=R_{ac}(t_1^{-1})R_{db}(t_2)$. Therefore, each ${\tilde A}_{v_n}^{g}\ot {\tilde A}_{V_N}^{g'}$ is a unitary operation on the space spanned by indices $u,v$.  

By combining the above argument with Eq.~\eqref{eq:eqspohHt}, we conclude that
\begin{equation}
\sigma_A^{(R,C)}:=\frac{1}{|G|^{n+N-2}d^2_{(R,C)}}\Pi_A(R,C)\,,
\end{equation}
where 
\begin{equation}
\Pi_A(R,C):=\sum{u,v}\sum_{\{h_k\}_{q_i}}\sum_{\{H_l\}_{q_j}}\left|(R,C);u,v;\{h_k\}_{q_i},\{H_l\}_{q_j}\right\>\left\<(R,C);u,v;\{h_k\}_{q_i},\{H_l\}_{q_j}\right|_A\,,
\end{equation}
is invariant under every ${\tilde A}_{v_i}^g$ and ${\tilde A}_{V_i}^g$. 
General $\sigma_A\in \Sigma(A)$ is written as a convex combination:
\begin{equation}\label{eq:infconvDG}
\sigma_A=\bigoplus_{a=(R,C)}p_{a}\sigma^a_A\,,
\end{equation}
which is consistent to Theorem~\ref{thm:convex}. 

\subsection{Calculation of \texorpdfstring{$\calI(A)_\Omega$}{I(A)Omega} for a thin annulus}
We are now ready to calculate an orthonormal basis of $\calK_{AB}$. We consider $A_+$, the support of all interaction terms nontrivially acting on spins in $A$. 
We decompose $B=B_{in}\cup B_{out}$ so that $B_{in}=A_+\cap D_{in}$ ($B_{out}=A_+\cap D_{out}$) contains spins around inner (outer) boundaries of $A$  (Fig.~\ref{fig:QDthick}). We choose the directions of the boundaries of $A_+$ in the same way as we did for $A$. Let us fix labels $\{h_i\}$ and $\{H_i\}$ in $A$. Spins in $B_{in}$, with the inner boundary of $A$, form another annulus. 
By requiring to be a $+1$ eigenstate of all $B_p$ acting on $B_{in}$, we can label states in $\calH_{B_{in}}$ by
\begin{equation}
|\{s_i\}, \{h_i\},t_1\>_{B_{in}}\,, 
\end{equation}
where $\{s_i\}$ corresponds to the inner boundary of $A_+$ (thick black circle in Fig.~\ref{fig:QDthick}), $t_1$ is the spin next to $t$. Other spins in $B_{in}$ are specified by fixing $\{h_i\}$ by the same reason we did for $\{g_i\}$ in $A$. Note that $\{h_i\}$ is not included in $B_{in}$, but needed to uniquely specify a vector in $\calH_{B_{in}}$. We repeat the same argument on $B_{out}$ to denote a vector  by 
\begin{equation}
|\{o_i\}, \{H_i\},t_2\>_{B_{out}}\,,
\end{equation}
where $\{o_i\}$ labels the outer boundary of $A_+$ after removing the corner effects by further requiring to be a $+1$ eigenstate of all $A_v$ on $B_{out}$. 
The total charges of $\{s_i\}$ and $\{o_i\}$ are constrained by $s=t_1ht_1^{-1}$ and $H=t_2ot_2^{-1}$, where $s=s_1...s_{n'}$ and $o=o_1...o_{N'}$. 
$h, H\in C$ implies that $s$ and $o$ are also in the same conjugacy class $C$. 
By using these notation, we can define $\left|(R,C), u,v,\{s_i\}_{q_i}\{h_i\}_{p_i}\right\>_{B_{in}}$ and $\left|(R,C),u,v,\{H_i\}_{p_j}\{o_i\}_{q_j}\right\>_{B_{out}}$ in the same way as in Eq.~\eqref{eq:RCbasis}.

A basis of $\calK_{AB}$ is given by 
\begin{align}
\left|(R,C);{u},{v};\{s_k\}_{q_i},\{o_k\}_{q_j}\right\>_{AB}:=&\sqrt{\frac{1}{|G|^{n+N-2}}}\sum_{p_i,p_j,c,d}\sum_{\{h_i\}_{p_i}}\sum_{\{H_i\}_{p_j}}
\left|(R,C), (q_i,p_i),(a,c),\{s_i\}_{q_i}\{h_i\}_{p_i}\right\>_{B_{in}}\\
&\ot\left|(R,C), (p_i,p_j),(c,d),\{h_i\}_{p_i}\{H_i\}_{p_j}\right\>_A\ot\left|(R,C), (p_j,q_j),(d,b),\{H_i\}_{p_j}\{o_i\}_{q_j}\right\>_{B_{out}}\,,
\end{align}
where $u=(q_i,q_j)$ and $v=(a,b)$. By construction, any basis element is a superposition of $+1$ eigenstates of all $B_p$ acting on $A_+$. 
It is also easy to check for $A_v^g$ by using the fact that they are combinations of permutations of terms in the superposition, and unitary rotations on $(c,d)$. 
There are $d_{(R,C)}^2$ choices of $(u,v)$ for fixed $(R,C)$, and there are $|G|^{n'+N'-2}$ choices of  $\{s_i\}_{q_i}$ and $\{o_i\}_{q_j}$ for fixed $u$ and $v$, and thus in total $\dim\calK_{AB}=|G|^{n'+N'}$. 
The vacuum sector of $D(G)$ corresponding to the label $(R,C)=(id, \{e\})$ has dimension $|G|^{n'+N'-2}$, since $d_{(id,\{e\})}=1$. 
Therefore 
\begin{equation}
\calI(A)_\Omega=\log\frac{d_\Omega^1}{d_{\Omega}}=-\log(|G|^{n'+N'-2}/|G|^{n'+N'})=\log|G|^2,
\end{equation}
which completes the proof.

\begin{figure}[htbp] 
\begin{center}
\includegraphics[width=6.0cm]{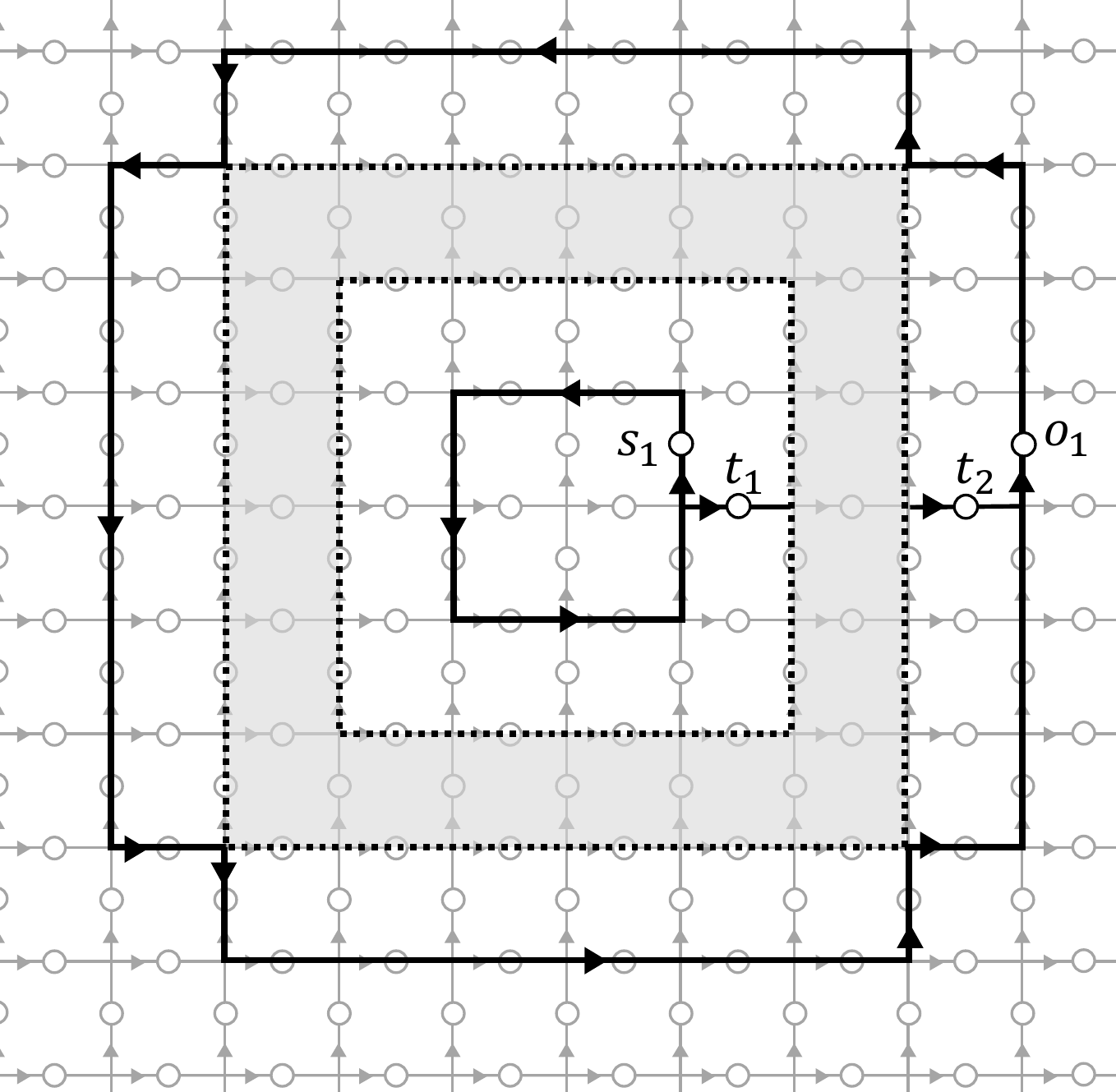}
\end{center}
\vspace{-4mm}
\caption{Region $A_+$ (surrounded by thick lines) includes all supports of interaction terms overlapping with $A$ (dotted region). We can label the basis of the ground subspace of $H_{A_+}$ as in a similar way as we did for $A$.}
\label{fig:QDthick}
\end{figure}

\begin{remark} The argument can be straightforwardly generalized to annuli containing only smooth boundaries. 
The boundaries of general annulus are mixtures of smooth and rough boundaries.  
One can still label the basis by using  $\{h_i\}$ and $\{H_j\}$ which corresponds to the (coarse-grained) boundaries. 
To do so, first choose the largest subregion of $A$ formed by only squares, which we denote by $A_-$. 
We can label the $+1$ eigenstates of $B_p$ and $A_v$ within $A_-$ by $|\{h_i\},\{H_j\},t\>_{A_-}$ as in the same way as in the case of the thinnest annulus (here, $t=t_1...t_k$ labels the total charge of the sites along a line crossing the annulus). All other degrees of freedom are specified by the restrictions. Then, we consider a product basis spanned by  $|\{h_i\},\{H_j\},t\>_{A_-}\ot|f_1...f_m\>_{A\backslash A_-}$, where $|f_1...f_m\>$ is a vector in the space of $A\backslash A_-$.  
Each $|f_k\>$ is in a support of some $A_v$ on $A$. Suppose $A_v$ acts on $|h_1\>, |h_2\>, |f_1\>$ and a site on $A_-$ which is already fixed and will be omitted.   
By applying $A_v^g$, the labels change to $h_1g^{-1}, gh_2$ and $gf_1$. Therefore, $h_1h_2$ and $h_1f_1$ are preserved under the action of $A_v$.  
Define a new basis by taking the equal weight superposition  over all $h_1,h_2,f_1$ satisfying $h_1h_2=h_1'$ and $h_1f_1= f_1'$. 
We denote the new basis by $|\{h_i\},\{H_j\},t\>_A$ after redefining $h_1\equiv h_1'$ and $h_2\equiv f_1'$. Each element of this new basis is a $+1$ eigenstate of $A_v$. 
By repeating this argument, we can specify a basis of the ground subspace by $|\{h_i\},\{H_j\},t\>_A$. 
\end{remark}

\section{Fibonacci anyons}\label{app:fibonacci}
The way the invariant in Eq.~\eqref{eq:InvGNSdim} is defined makes clear that it is always the logarithm of a rational number.
However, it is known that there are anyon models with \emph{irrational} quantum dimension.
Perhaps the best known example are the Fibonacci anyons~\cite{Preskill}, which is similar to the Yang-Lee model.
This raises the obvious question: can the invariant we define capture such phases?
More precisely, we are interested in models where the \emph{square} of the quantum dimensions is not an integer (compare with Theorem~\ref{thm:qdouble}).
Models for which the square of the quantum dimensions are integers are called \emph{weakly integral}, and include quantum double models coming from groups (including the so-called \emph{twisted} quantum double models)~\cite{CuiHW2015}.
Unfortunately we do not have direct answer to this.
We will however here consider a slightly different setting, where we can define a similar quantity, and will comment on how it relates to the definition of $\mathcal{I}(A)_\Omega$.

Consider the Fibonacci anyon model, which has only one non-trivial anyon $\tau$, with fusion rule $\tau \otimes \tau = \iota \oplus \tau$.
Note that in particular the model is non-abelian and that $\tau$ is self-dual.
Now consider a chain of $n$ $\tau$-anyons, whose combined charge is trivial, in the sense that the $n$ anyons together fuse to the trivial sector. 
We will group the anyons in two blocks, $n = n_A + n_B$, where we picture the anyons on a line, with the $n_A$ leftmost anyons belonging to group $A$, and $n_B$ $\tau$-anyons on the right of that.
A basis for the state space of such a configuration is most conveniently described by \emph{fusion trees}, which in the category theory picture correspond to morphisms in $\operatorname{Hom}(\tau^{\otimes n}, \iota)$ in the language of tensor categories~\footnote{Equivalently one could look at the space $\operatorname{Hom}(\iota, \tau^{\otimes n})$, which could be interpreted as all the ways to create $\tau$ anyons out of the vacuum.}.
Readers who are not familiar with the category theoretical picture can consult the book by Wang~\cite{Wang}.
If $F(k)$ is the $k$-th Fibonacci number, it is easy to deduce by induction that $\dim \operatorname{Hom}(\tau^{\otimes n}, \iota) = F(n-1)$.
Similarly, for fusions to $\tau$, we have $\dim \operatorname{Hom}(\tau^{\otimes n}, \tau) = F(n)$.
This explains the name of the model.

Let $\mathcal{H}_n$ be the total Hilbert space of the system, which has dimension $F(n-1)$ as we have seen.
Operators on $\mathcal{H}_n$ can be represented using the graphical language of tensor categories.
In particular, a basis can be obtained by taking fusion trees, and pasting them together with a ``flipped'' fusion tree.
This gives a morphism in $\operatorname{End}(\tau^{\otimes n}) \equiv \operatorname{Hom}(\tau^{\otimes n}, \tau^{\otimes n})$.
Note that since the total charge of the system is trivial, we only have to consider those diagrams that go through a single $\iota$ line.
The algebra of all such operators, which can be identified with $\operatorname{End}(\tau^{\otimes n})$, will be denoted by $\mathcal{E}$.

We can now define projections onto the total charge in the regions $A$ and $B$, respectively.
We will call them $P^A_\iota$ and $P^A_\tau$, and similarly for $B$.
Again these can be represented by gluing together fusion trees with their (horizontally) flipped version, where one has to consider all trees that fuse to $\iota$ and $\tau$ respectively, and straight lines (i.e., identity morphisms) on the $B$ part.
Note that $P^A_\iota P_\tau^B = 0$, since a $\iota$ and a $\tau$ cannot fuse to the vacuum.
However, the projections $P_\iota \equiv P^A_\iota P^B \iota$ and $P_\tau \equiv P^A_\tau P^B_\tau$ are non-trivial, mutually commuting, and sum up to the identity.

We say that an operation is \emph{local} with respect to the bipartition $AB$ if it does not change the total charge in either the $A$ or $B$ region (note that it is not possible to change the charge in only one region, since all anyons together have to fuse to the vacuum, hence the total charge in $A$ and $B$ must be either both $\iota$, or both $\tau$). 
This leads to the algebra
\[
	\mathcal{C} \equiv \{P^A_\iota P^B_\iota, P^A_\tau P^B_\tau\}' \cap \mathcal{E} = P_\iota \mathcal{E} P_\iota \oplus P_\tau \mathcal{E} P_\tau \equiv \mathcal{C}_\iota \oplus \mathcal{C}_\tau.
\]
Note the similarity with equation~\eqref{eq:logicalC}: again we have a decomposition into superselection sectors.

\begin{remark}
We do not need to divide out the ideal $\cal{N}$, or use the idea of the annulus, since we already are working on the level of charges here.
This is a key difference with the approach we used above, where a key step is to identify the states which a certain total charge within the annulus. 
\end{remark}

First consider the algebra $\mathcal{C}_\iota$.
This algebra is generated by diagrams in $\operatorname{End}(\tau^{\otimes n})$ which only act non-trivially on the first $n_A$ anyons, and similar diagrams acting only on $B$.
The total subspace of states that have charge $\iota$ in both regions $A$ and $B$ has dimension $F(n_A-1) F(n_B-1)$.
It follows that $\calC_\iota \cong M_{F(n_A-1)F(n_B-1)}(\mathbb{C})$.
Similarly, the dimension of the space fusing to $\tau \otimes \tau$ charges fusing to $\iota$ is $F(n_A) F(n_B)$, and $\calC_\tau \cong M_{F(n_A) F(n_B)}$.
As a consistency check, note that
\[
	F(n_A-1) F(n_B -1) + F(n_A) F(n_B) = F(n_A+n_B-1),
\]
so the dimensions match up.
This equation can be verified by using repeatedly that $F(k-1)F(l-1) + F(k)F(l) = F(k-2)F(l) + F(k-1)F(l+1)$.

It turns out that we can recover the quantum dimensions by comparing the size of the algebra for each sector with the size of the algebra of the trivial sector.
More precisely, let $\sigma$ be a tracial state on $\calE$. Note that since $\calE$ is irreducible, this is unique and coincides with the usual trace of a matrix algebra.
Note that the sectors are obtained by cutting down $\calE$ with projections $P_\iota$ and $P_\tau$.
Hence we can look at the ratios $\sigma(P_k)/\sigma(P_\iota)$ to compare the sizes of the different algebras.
Since the quantum dimension for the $\tau$-anyon is not rational, it is particularly interesting to consider the limit where both $n_A$ and $n_B$ go to infinity.
For the $\tau$-sector, this yields
\[
	\lim_{n_A, n_B \to \infty} \frac{\tau(P^{AB}_\tau)}{\tau(P^{AB}_\iota)} = \lim_{n_A, n_B \to \infty} \frac{F(n_A) F(n_B)}{F(n_A-1) F(n_B-1)} = \phi^2,
\]
where $\phi$ is the golden ratio.
Here we used that $\lim_{n \to \infty} \frac{F(n+1)}{F(n)} = \phi$.
Similarly, $\sigma(I)/\sigma(P^{AB}_\iota)$ tends to $1 + \varphi^2$, the total quantum dimension of the theory.
Note that this is very similar to Theorem~\ref{thm:InvGNSdim}, up to a logarithm.
However, here the step of identifying the correct central projections (or, equivalently, the correct subspaces) is much more direct, since we directly work with the ``internal'' fusion trees.
Since we work directly with the charges, the question of stability under perturbations does not directly apply here.
This is because we do not talk about how to get the anyon structure (as a modular tensor category) from the underlying physical model.
That is precisely what is non-trivial in showing that quantity is stable.

This simple example illustrates a pathway to obtain irrational quantum dimensions.
Since we are discussing finite dimensional systems, it is clear that some limit procedure has to be involved.
Translating the Fibonacci example back to our original setting suggests that one has to consider limits where the inside (and outside) of the annulus can contain more and more excitations.
Hence one either has to take a limit of growing system size and growing annuli, or keep the annulus fixed and increase the number of sites inside the annulus.
The latter essentially means that one has to rescale the distance between the sites inside the annulus, hence both ways are essentially the same.
This makes the stability argument more subtle, since one needs a ``stable logical algebra'' condition for whole sequence of anyon configurations.
Finally, note that in the Fibonacci example there is only one non-trivial superselection sector.
Hence in the general case, one has to consider all possible ways that anyons in region $A$ (not necessarily of the same type!) can fuse to one of the charges, and to the conjugate charge for the remaining anyons in region $B$.
We leave this analysis open for future work.

\section{Comparison to thermodynamic limit}\label{sec:thermodynamic}
The different superselection sectors appear in the analysis in Section~\ref{sec:setting} as different blocks in the decomposition of $\calE/\calN$. 
This decomposition is key in identifying the different types of anyons the system has.
The analysis is partly motivated by the study of superselection sectors in the thermodynamic limit, where an operator algebraic approach is used.
There one can study the anyons (and all their properties such as braiding) directly in the thermodynamic limit.
It is therefore instructive to compare these different frameworks.
A brief introduction to this approach can be found in Ref.~\cite{FiedlerNO17}.

The main idea behind this approach is that superselection sectors can be identified with (equivalence classes of) irreducible representations of the algebra of quasi-local observables, which is generated by all observables that act only on finitely many (but otherwise arbitrarily large) number of sites.
If two representations $\pi_1$ and $\pi_2$ are inequivalent, it can be shown that there is no unitary $U$ that maps a vector state $|\psi\>$ in one representation to a vector state in the other representation (in a way that is compatible with both representations).
Physically this can be understood as the impossibility to transform a state in one representation to a state in the other one with a local operation.
Or in the language of anyons: with local operations one cannot change the \emph{total} charge of the system.
Hence one recovers the notion of superselection sectors we have used earlier.

It is possible to give a more direction connection between the thermodynamic limit approach and the one we take here, at least in the case of abelian quantum double models.
In the thermodynamic limit of such models, one can calculate what is called the \emph{Jones index} of a certain inclusion of operator algebras, and show that it is equal to the total quantum dimension of the theory.
In Ref.~\cite{Naaijkens18} it is shown that this can in fact be obtained using a limiting procedure of finite dimensional algebras $\calR_i \subset \widehat{\calR}_i$, or more precisely, from an optimization of certain relative entropies related to these finite dimensional systems.
The finite dimensional algebras $\calR_i$ are generated by the operators supported on two disjoint patches separated by a large enough distance.
Note that this is quite similar to the annulus picture above, with the distinction that the ''annulus'' encloses \emph{two} regions.
This geometric configuration is more convenient once one wants to take the limit where the two isolated regions grow to infinity.
The algebra $\calR_i$ can then be understood as the algebra of all operations \emph{on the two patches} that leave the charge in both regions invariant.
It does however \emph{not} decompose into different sectors, but rather corresponds to the $\calC_1$ component from Eq.~\eqref{eq:logicalC}.

The algebra $\widehat{\calR}_i$ is then generated by $\calR_i$, and operators which create a pair of excitations in each of the patches (without creating any excitations outside of the two patches).
Hence this algebra can be compared to $\calE$ of  Eq.~\eqref{def:algE}.
For the toric code, it can be generated by $\calR_i$, together with a string operator of each type connecting the two regions.
This algebra can then be decomposed into four sectors, corresponding to the different anyon types in the toric code.
Moreover, $\calR_i$ embeds into $\widehat{\calR}_i$ as $A \mapsto A \oplus A \oplus A \oplus A$.
We expect a similar structure to be true at least for any abelian model.

The advantage is that $\widehat{\calR}$ (and their finite dimensional approximations $\widehat{\calR}_i$) in the thermodynamic limit can be defined in a purely algebraic way, without any direct reference to the Hamiltonian, as the algebra of all operators that commute with all local observables localized outside of the two cones.
The algebra \emph{does} however depend on the Hamiltonian in a subtle way: to define it, one first has to represent the algebra of observables of the system on a Hilbert space. This is done by taking a ground state of the system (which can be defined abstractly), and looking at the corresponding GNS representation. This algebra $\widehat{\calR}$ depends very much on this representation.
It should be noted however that it is not necessary to make any assumptions on the Hamiltonian, in particular we don't have to restrict to LCPC Hamiltonians.

Since we are particularly interested in \emph{phases} of matter (which are further explained in Section~\ref{sec:entropic}), it is important to understand what happens after perturbing the Hamiltonian.
If the perturbation is small enough that it does not close the gap, the ground state of the perturbed model is in the same phase.
Hence, one would expect that the superselection structure of the perturbed model is the same.
It can be shown that under natural assumptions, this is indeed the case, and the perturbed model has for example the same sectors and the anyons have the same braiding properties~\cite{ChaNN18}.
It is however less clear how certain von Neumann algebraic aspects behave under perturbations.
In particular, this is true for the Jones index of the inclusion $\calR \subset \widehat{\calR}$ which we mentioned above.
For certain unperturbed models (such as the toric code), this can be calculated explicitly and be shown to be equal to the the total quantum dimension of the theory.
However, even though we know that the superselection structure is invariant under perturbations, we presently have no control over the Jones index.
Hence our understanding is far from satisfactory at the moment.
This is even more so the case for the topological entanglement entropy, of which we do not have a satisfactory understanding in the infinite sytem size setting.
Indeed, this is one of the motivations behind the present work.

%


\begin{thebibliography}{50}%
\makeatletter
\providecommand \@ifxundefined [1]{%
 \@ifx{#1\undefined}
}%
\providecommand \@ifnum [1]{%
 \ifnum #1\expandafter \@firstoftwo
 \else \expandafter \@secondoftwo
 \fi
}%
\providecommand \@ifx [1]{%
 \ifx #1\expandafter \@firstoftwo
 \else \expandafter \@secondoftwo
 \fi
}%
\providecommand \natexlab [1]{#1}%
\providecommand \enquote  [1]{``#1''}%
\providecommand \bibnamefont  [1]{#1}%
\providecommand \bibfnamefont [1]{#1}%
\providecommand \citenamefont [1]{#1}%
\providecommand \href@noop [0]{\@secondoftwo}%
\providecommand \href [0]{\begingroup \@sanitize@url \@href}%
\providecommand \@href[1]{\@@startlink{#1}\@@href}%
\providecommand \@@href[1]{\endgroup#1\@@endlink}%
\providecommand \@sanitize@url [0]{\catcode `\\12\catcode `\$12\catcode
  `\&12\catcode `\#12\catcode `\^12\catcode `\_12\catcode `\%12\relax}%
\providecommand \@@startlink[1]{}%
\providecommand \@@endlink[0]{}%
\providecommand \url  [0]{\begingroup\@sanitize@url \@url }%
\providecommand \@url [1]{\endgroup\@href {#1}{\urlprefix }}%
\providecommand \urlprefix  [0]{URL }%
\providecommand \Eprint [0]{\href }%
\providecommand \doibase [0]{http://dx.doi.org/}%
\providecommand \selectlanguage [0]{\@gobble}%
\providecommand \bibinfo  [0]{\@secondoftwo}%
\providecommand \bibfield  [0]{\@secondoftwo}%
\providecommand \translation [1]{[#1]}%
\providecommand \BibitemOpen [0]{}%
\providecommand \bibitemStop [0]{}%
\providecommand \bibitemNoStop [0]{.\EOS\space}%
\providecommand \EOS [0]{\spacefactor3000\relax}%
\providecommand \BibitemShut  [1]{\csname bibitem#1\endcsname}%
\let\auto@bib@innerbib\@empty
\bibitem [{\citenamefont {Chen}\ \emph {et~al.}(2010)\citenamefont {Chen},
  \citenamefont {Gu},\ and\ \citenamefont {Wen}}]{ChenGW10}%
  \BibitemOpen
  \bibfield  {author} {\bibinfo {author} {\bibfnamefont {Xie}\ \bibnamefont
  {Chen}}, \bibinfo {author} {\bibfnamefont {Zheng-Cheng}\ \bibnamefont {Gu}},
  \ and\ \bibinfo {author} {\bibfnamefont {Xiao-Gang}\ \bibnamefont {Wen}},\
  }\bibfield  {title} {\enquote {\bibinfo {title} {Local unitary
  transformation, long-range quantum entanglement, wave function
  renormalization, and topological order},}\ }\href {\doibase
  10.1103/PhysRevB.82.155138} {\bibfield  {journal} {\bibinfo  {journal} {Phys.
  Rev. B}\ }\textbf {\bibinfo {volume} {82}},\ \bibinfo {pages} {155138}
  (\bibinfo {year} {2010})},\ \Eprint {http://arxiv.org/abs/1004.3835}
  {arXiv:1004.3835 [cond-mat]} \BibitemShut {NoStop}%
\bibitem [{\citenamefont {Wen}(1989)}]{Wen89}%
  \BibitemOpen
  \bibfield  {author} {\bibinfo {author} {\bibfnamefont {X.~G.}\ \bibnamefont
  {Wen}},\ }\bibfield  {title} {\enquote {\bibinfo {title} {Vacuum degeneracy
  of chiral spin states in compactified space},}\ }\href {\doibase
  10.1103/PhysRevB.40.7387} {\bibfield  {journal} {\bibinfo  {journal} {Phys.
  Rev. B}\ }\textbf {\bibinfo {volume} {40}},\ \bibinfo {pages} {7387--7390}
  (\bibinfo {year} {1989})}\BibitemShut {NoStop}%
\bibitem [{\citenamefont {Kitaev}(2003)}]{Kitaev03}%
  \BibitemOpen
  \bibfield  {author} {\bibinfo {author} {\bibfnamefont {Alexei}\ \bibnamefont
  {Kitaev}},\ }\bibfield  {title} {\enquote {\bibinfo {title} {Fault-tolerant
  quantum computation by anyons},}\ }\href {\doibase
  10.1016/S0003-4916(02)00018-0} {\bibfield  {journal} {\bibinfo  {journal}
  {Ann. Physics}\ }\textbf {\bibinfo {volume} {303}},\ \bibinfo {pages} {2--30}
  (\bibinfo {year} {2003})},\ \Eprint {http://arxiv.org/abs/quant-ph/9707021}
  {arXiv:quant-ph/9707021 [quant-ph]} \BibitemShut {NoStop}%
\bibitem [{\citenamefont {Michael H.~Freedman}\ and\ \citenamefont
  {Wang}(2003)}]{Freedman03}%
  \BibitemOpen
  \bibfield  {author} {\bibinfo {author} {\bibfnamefont {Michael J.~Larsen}\
  \bibnamefont {Michael H.~Freedman}, \bibfnamefont {Alexei~Kitaev}}\ and\
  \bibinfo {author} {\bibfnamefont {Zhenghan}\ \bibnamefont {Wang}},\
  }\bibfield  {title} {\enquote {\bibinfo {title} {Topological quantum
  computation},}\ }\href {\doibase
  https://doi.org/10.1090/S0273-0979-02-00964-3} {\bibfield  {journal}
  {\bibinfo  {journal} {Bull. Amer. Math. Soc.}\ }\textbf {\bibinfo {volume}
  {40}},\ \bibinfo {pages} {31--38} (\bibinfo {year} {2003})}\BibitemShut
  {NoStop}%
\bibitem [{\citenamefont {Hastings}\ and\ \citenamefont
  {Koma}(2006)}]{HastingsK06}%
  \BibitemOpen
  \bibfield  {author} {\bibinfo {author} {\bibfnamefont {Matthew~B.}\
  \bibnamefont {Hastings}}\ and\ \bibinfo {author} {\bibfnamefont {Tohru}\
  \bibnamefont {Koma}},\ }\bibfield  {title} {\enquote {\bibinfo {title}
  {Spectral gap and exponential decay of correlations},}\ }\href {\doibase
  10.1007/s00220-006-0030-4} {\bibfield  {journal} {\bibinfo  {journal}
  {Commun. Math. Phys.}\ }\textbf {\bibinfo {volume} {265}},\ \bibinfo {pages}
  {781--804} (\bibinfo {year} {2006})}\BibitemShut {NoStop}%
\bibitem [{\citenamefont {Nachtergaele}\ and\ \citenamefont
  {Sims}(2006)}]{NachtergaeleS06}%
  \BibitemOpen
  \bibfield  {author} {\bibinfo {author} {\bibfnamefont {Bruno}\ \bibnamefont
  {Nachtergaele}}\ and\ \bibinfo {author} {\bibfnamefont {Robert}\ \bibnamefont
  {Sims}},\ }\bibfield  {title} {\enquote {\bibinfo {title} {Lieb-{R}obinson
  bounds and the exponential clustering theorem},}\ }\href {\doibase
  10.1007/s00220-006-1556-1} {\bibfield  {journal} {\bibinfo  {journal}
  {Commun. Math. Phys.}\ }\textbf {\bibinfo {volume} {265}},\ \bibinfo {pages}
  {119--130} (\bibinfo {year} {2006})}\BibitemShut {NoStop}%
\bibitem [{\citenamefont {Levin}\ and\ \citenamefont {Wen}(2006)}]{LevinW06}%
  \BibitemOpen
  \bibfield  {author} {\bibinfo {author} {\bibfnamefont {Michael}\ \bibnamefont
  {Levin}}\ and\ \bibinfo {author} {\bibfnamefont {Xiao-Gang}\ \bibnamefont
  {Wen}},\ }\bibfield  {title} {\enquote {\bibinfo {title} {Detecting
  topological order in a ground state wave function},}\ }\href {\doibase
  10.1103/PhysRevLett.96.110405} {\bibfield  {journal} {\bibinfo  {journal}
  {Phys. Rev. Lett.}\ }\textbf {\bibinfo {volume} {96}},\ \bibinfo {pages}
  {110405} (\bibinfo {year} {2006})},\ \Eprint
  {http://arxiv.org/abs/cond-mat/0510613} {arXiv:cond-mat/0510613 [cond-mat]}
  \BibitemShut {NoStop}%
\bibitem [{\citenamefont {Hastings}\ and\ \citenamefont
  {Wen}(2005)}]{PhysRevB.72.045141}%
  \BibitemOpen
  \bibfield  {author} {\bibinfo {author} {\bibfnamefont {M.~B.}\ \bibnamefont
  {Hastings}}\ and\ \bibinfo {author} {\bibfnamefont {Xiao-Gang}\ \bibnamefont
  {Wen}},\ }\bibfield  {title} {\enquote {\bibinfo {title} {Quasiadiabatic
  continuation of quantum states: The stability of topological ground-state
  degeneracy and emergent gauge invariance},}\ }\href {\doibase
  10.1103/PhysRevB.72.045141} {\bibfield  {journal} {\bibinfo  {journal} {Phys.
  Rev. B}\ }\textbf {\bibinfo {volume} {72}},\ \bibinfo {pages} {045141}
  (\bibinfo {year} {2005})}\BibitemShut {NoStop}%
\bibitem [{\citenamefont {Zhang}\ \emph {et~al.}(2011)\citenamefont {Zhang},
  \citenamefont {Grover},\ and\ \citenamefont {Vishwanath}}]{Zhang11}%
  \BibitemOpen
  \bibfield  {author} {\bibinfo {author} {\bibfnamefont {Yi}~\bibnamefont
  {Zhang}}, \bibinfo {author} {\bibfnamefont {Tarun}\ \bibnamefont {Grover}}, \
  and\ \bibinfo {author} {\bibfnamefont {Ashvin}\ \bibnamefont {Vishwanath}},\
  }\bibfield  {title} {\enquote {\bibinfo {title} {Topological entanglement
  entropy of ${\mathbb{z}}_{2}$ spin liquids and lattice laughlin states},}\
  }\href {\doibase 10.1103/PhysRevB.84.075128} {\bibfield  {journal} {\bibinfo
  {journal} {Phys. Rev. B}\ }\textbf {\bibinfo {volume} {84}},\ \bibinfo
  {pages} {075128} (\bibinfo {year} {2011})}\BibitemShut {NoStop}%
\bibitem [{\citenamefont {{Isakov}}\ \emph {et~al.}(2011)\citenamefont
  {{Isakov}}, \citenamefont {{Hastings}},\ and\ \citenamefont
  {{Melko}}}]{Isakov11}%
  \BibitemOpen
  \bibfield  {author} {\bibinfo {author} {\bibfnamefont {S.~V.}\ \bibnamefont
  {{Isakov}}}, \bibinfo {author} {\bibfnamefont {M.~B.}\ \bibnamefont
  {{Hastings}}}, \ and\ \bibinfo {author} {\bibfnamefont {R.~G.}\ \bibnamefont
  {{Melko}}},\ }\bibfield  {title} {\enquote {\bibinfo {title} {{Topological
  entanglement entropy of a Bose-Hubbard spin liquid}},}\ }\href {\doibase
  10.1038/nphys2036} {\bibfield  {journal} {\bibinfo  {journal} {Nature
  Physics}\ }\textbf {\bibinfo {volume} {7}},\ \bibinfo {pages} {772--775}
  (\bibinfo {year} {2011})},\ \Eprint {http://arxiv.org/abs/1102.1721}
  {arXiv:1102.1721 [cond-mat.str-el]} \BibitemShut {NoStop}%
\bibitem [{\citenamefont {{Jiang}}\ \emph {et~al.}(2012)\citenamefont
  {{Jiang}}, \citenamefont {{Wang}},\ and\ \citenamefont
  {{Balents}}}]{Jiang12}%
  \BibitemOpen
  \bibfield  {author} {\bibinfo {author} {\bibfnamefont {H.-C.}\ \bibnamefont
  {{Jiang}}}, \bibinfo {author} {\bibfnamefont {Z.}~\bibnamefont {{Wang}}}, \
  and\ \bibinfo {author} {\bibfnamefont {L.}~\bibnamefont {{Balents}}},\
  }\bibfield  {title} {\enquote {\bibinfo {title} {{Identifying topological
  order by entanglement entropy}},}\ }\href {\doibase 10.1038/nphys2465}
  {\bibfield  {journal} {\bibinfo  {journal} {Nature Physics}\ }\textbf
  {\bibinfo {volume} {8}},\ \bibinfo {pages} {902--905} (\bibinfo {year}
  {2012})},\ \Eprint {http://arxiv.org/abs/1205.4289} {arXiv:1205.4289
  [cond-mat.str-el]} \BibitemShut {NoStop}%
\bibitem [{\citenamefont {Kitaev}\ and\ \citenamefont
  {Preskill}(2006)}]{KitaevP06}%
  \BibitemOpen
  \bibfield  {author} {\bibinfo {author} {\bibfnamefont {Alexei}\ \bibnamefont
  {Kitaev}}\ and\ \bibinfo {author} {\bibfnamefont {John}\ \bibnamefont
  {Preskill}},\ }\bibfield  {title} {\enquote {\bibinfo {title} {Topological
  entanglement entropy},}\ }\href {\doibase 10.1103/PhysRevLett.96.110404}
  {\bibfield  {journal} {\bibinfo  {journal} {Phys. Rev. Lett.}\ }\textbf
  {\bibinfo {volume} {96}},\ \bibinfo {pages} {110404} (\bibinfo {year}
  {2006})},\ \Eprint {http://arxiv.org/abs/hep-th/0510092}
  {arXiv:hep-th/0510092 [hep-th]} \BibitemShut {NoStop}%
\bibitem [{\citenamefont {Hamma}\ \emph {et~al.}(2005)\citenamefont {Hamma},
  \citenamefont {Ionicioiu},\ and\ \citenamefont
  {Zanardi}}]{PhysRevA.71.022315}%
  \BibitemOpen
  \bibfield  {author} {\bibinfo {author} {\bibfnamefont {Alioscia}\
  \bibnamefont {Hamma}}, \bibinfo {author} {\bibfnamefont {Radu}\ \bibnamefont
  {Ionicioiu}}, \ and\ \bibinfo {author} {\bibfnamefont {Paolo}\ \bibnamefont
  {Zanardi}},\ }\bibfield  {title} {\enquote {\bibinfo {title} {Bipartite
  entanglement and entropic boundary law in lattice spin systems},}\ }\href
  {\doibase 10.1103/PhysRevA.71.022315} {\bibfield  {journal} {\bibinfo
  {journal} {Phys. Rev. A}\ }\textbf {\bibinfo {volume} {71}},\ \bibinfo
  {pages} {022315} (\bibinfo {year} {2005})}\BibitemShut {NoStop}%
\bibitem [{\citenamefont {Zou}\ and\ \citenamefont
  {Haah}(2016)}]{PhysRevB.94.075151}%
  \BibitemOpen
  \bibfield  {author} {\bibinfo {author} {\bibfnamefont {L.}~\bibnamefont
  {Zou}}\ and\ \bibinfo {author} {\bibfnamefont {J.}~\bibnamefont {Haah}},\
  }\bibfield  {title} {\enquote {\bibinfo {title} {Spurious long-range
  entanglement and replica correlation length},}\ }\href {\doibase
  10.1103/PhysRevB.94.075151} {\bibfield  {journal} {\bibinfo  {journal} {Phys.
  Rev. B}\ }\textbf {\bibinfo {volume} {94}},\ \bibinfo {pages} {075151}
  (\bibinfo {year} {2016})}\BibitemShut {NoStop}%
\bibitem [{\citenamefont {{Williamson}}\ \emph {et~al.}(2019)\citenamefont
  {{Williamson}}, \citenamefont {{Dua}},\ and\ \citenamefont
  {{Cheng}}}]{Williamson18}%
  \BibitemOpen
  \bibfield  {author} {\bibinfo {author} {\bibfnamefont {D.~J.}\ \bibnamefont
  {{Williamson}}}, \bibinfo {author} {\bibfnamefont {A.}~\bibnamefont {{Dua}}},
  \ and\ \bibinfo {author} {\bibfnamefont {M.}~\bibnamefont {{Cheng}}},\
  }\bibfield  {title} {\enquote {\bibinfo {title} {{Spurious topological
  entanglement entropy from subsystem symmetries}},}\ }\href {\doibase
  10.1103/PhysRevLett.122.140506} {\bibfield  {journal} {\bibinfo  {journal}
  {Phys. Rev. Lett.}\ }\textbf {\bibinfo {volume} {122}},\ \bibinfo {pages}
  {140506} (\bibinfo {year} {2019})},\ \Eprint
  {http://arxiv.org/abs/1808.05221} {arXiv:1808.05221 [quant-ph]} \BibitemShut
  {NoStop}%
\bibitem [{\citenamefont {Kennedy}\ and\ \citenamefont
  {Tasaki}(1992)}]{kennedy1992}%
  \BibitemOpen
  \bibfield  {author} {\bibinfo {author} {\bibfnamefont {Tom}\ \bibnamefont
  {Kennedy}}\ and\ \bibinfo {author} {\bibfnamefont {Hal}\ \bibnamefont
  {Tasaki}},\ }\bibfield  {title} {\enquote {\bibinfo {title} {Hidden symmetry
  breaking and the haldane phase in $s=1$ quantum spin chains},}\ }\href
  {https://projecteuclid.org:443/euclid.cmp/1104250747} {\bibfield  {journal}
  {\bibinfo  {journal} {Comm. Math. Phys.}\ }\textbf {\bibinfo {volume}
  {147}},\ \bibinfo {pages} {431--484} (\bibinfo {year} {1992})}\BibitemShut
  {NoStop}%
\bibitem [{\citenamefont {Haah}(2016)}]{Haah16}%
  \BibitemOpen
  \bibfield  {author} {\bibinfo {author} {\bibfnamefont {Jeongwan}\
  \bibnamefont {Haah}},\ }\bibfield  {title} {\enquote {\bibinfo {title} {An
  invariant of topologically ordered states under local unitary
  transformations},}\ }\href {\doibase 10.1007/s00220-016-2594-y} {\bibfield
  {journal} {\bibinfo  {journal} {Commun. Math. Phys.}\ }\textbf {\bibinfo
  {volume} {342}},\ \bibinfo {pages} {771--801} (\bibinfo {year} {2016})},\
  \Eprint {http://arxiv.org/abs/1407.2926} {arXiv:1407.2926 [quant-ph]}
  \BibitemShut {NoStop}%
\bibitem [{\citenamefont {Verlinde}(1988)}]{Verlinde88}%
  \BibitemOpen
  \bibfield  {author} {\bibinfo {author} {\bibfnamefont {Erik}\ \bibnamefont
  {Verlinde}},\ }\bibfield  {title} {\enquote {\bibinfo {title} {Fusion rules
  and modular transformations in {$2$}{D} conformal field theory},}\ }\href
  {\doibase 10.1016/0550-3213(88)90603-7} {\bibfield  {journal} {\bibinfo
  {journal} {Nuclear Phys. B}\ }\textbf {\bibinfo {volume} {300}},\ \bibinfo
  {pages} {360--376} (\bibinfo {year} {1988})}\BibitemShut {NoStop}%
\bibitem [{\citenamefont {Wang}(2010)}]{Wang}%
  \BibitemOpen
  \bibfield  {author} {\bibinfo {author} {\bibfnamefont {Zhenghan}\
  \bibnamefont {Wang}},\ }\href@noop {} {\emph {\bibinfo {title} {Topological
  Quantum Computation}}},\ \bibinfo {series} {CBMS Regional Conference Series
  in Mathematics}, Vol.\ \bibinfo {volume} {112}\ (\bibinfo  {publisher}
  {Published for the Conference Board of the Mathematical Sciences, Washington,
  DC},\ \bibinfo {year} {2010})\ pp.\ \bibinfo {pages} {xiii+115}\BibitemShut
  {NoStop}%
\bibitem [{\citenamefont {Levin}\ and\ \citenamefont {Wen}(2005)}]{LevinW05}%
  \BibitemOpen
  \bibfield  {author} {\bibinfo {author} {\bibfnamefont {Michael~A.}\
  \bibnamefont {Levin}}\ and\ \bibinfo {author} {\bibfnamefont {Xiao-Gang}\
  \bibnamefont {Wen}},\ }\bibfield  {title} {\enquote {\bibinfo {title}
  {String-net condensation: A physical mechanism for topological phases},}\
  }\href {\doibase 10.1103/PhysRevB.71.045110} {\bibfield  {journal} {\bibinfo
  {journal} {Phys. Rev. B}\ }\textbf {\bibinfo {volume} {71}},\ \bibinfo
  {pages} {045110} (\bibinfo {year} {2005})},\ \Eprint
  {http://arxiv.org/abs/cond-mat/0404617} {arXiv:cond-mat/0404617 [cond-mat]}
  \BibitemShut {NoStop}%
\bibitem [{\citenamefont {Michalakis}\ and\ \citenamefont
  {Zwolak}(2013)}]{MichalakisZ13}%
  \BibitemOpen
  \bibfield  {author} {\bibinfo {author} {\bibfnamefont {Spyridon}\
  \bibnamefont {Michalakis}}\ and\ \bibinfo {author} {\bibfnamefont
  {Justyna~P.}\ \bibnamefont {Zwolak}},\ }\bibfield  {title} {\enquote
  {\bibinfo {title} {Stability of frustration-free {H}amiltonians},}\ }\href
  {\doibase 10.1007/s00220-013-1762-6} {\bibfield  {journal} {\bibinfo
  {journal} {Commun. Math. Phys.}\ }\textbf {\bibinfo {volume} {322}},\
  \bibinfo {pages} {277--302} (\bibinfo {year} {2013})},\ \Eprint
  {http://arxiv.org/abs/1109.1588} {arXiv:1109.1588 [quant-ph]} \BibitemShut
  {NoStop}%
\bibitem [{\citenamefont {Bravyi}\ \emph {et~al.}(2010)\citenamefont {Bravyi},
  \citenamefont {Hastings},\ and\ \citenamefont {Michalakis}}]{BravyiHM10}%
  \BibitemOpen
  \bibfield  {author} {\bibinfo {author} {\bibfnamefont {Sergey}\ \bibnamefont
  {Bravyi}}, \bibinfo {author} {\bibfnamefont {Matthew~B.}\ \bibnamefont
  {Hastings}}, \ and\ \bibinfo {author} {\bibfnamefont {Spyridon}\ \bibnamefont
  {Michalakis}},\ }\bibfield  {title} {\enquote {\bibinfo {title} {Topological
  quantum order: Stability under local perturbations},}\ }\href {\doibase
  10.1063/1.3490195} {\bibfield  {journal} {\bibinfo  {journal} {J. Math.
  Phys.}\ }\textbf {\bibinfo {volume} {51}},\ \bibinfo {pages} {093512}
  (\bibinfo {year} {2010})},\ \Eprint {http://arxiv.org/abs/1001.0344}
  {arXiv:1001.0344 [quant-ph]} \BibitemShut {NoStop}%
\bibitem [{\citenamefont {{Bravyi}}\ and\ \citenamefont
  {{Hastings}}(2011)}]{BravyiH11}%
  \BibitemOpen
  \bibfield  {author} {\bibinfo {author} {\bibfnamefont {S.}~\bibnamefont
  {{Bravyi}}}\ and\ \bibinfo {author} {\bibfnamefont {M.~B.}\ \bibnamefont
  {{Hastings}}},\ }\bibfield  {title} {\enquote {\bibinfo {title} {A short
  proof of stability of topological order under local perturbations},}\
  }\href@noop {} {\bibfield  {journal} {\bibinfo  {journal} {Comm. Math.
  Phys.}\ }\textbf {\bibinfo {volume} {307}},\ \bibinfo {pages} {609--627}
  (\bibinfo {year} {2011})},\ \Eprint {http://arxiv.org/abs/1001.4363}
  {arXiv:1001.4363 [math-ph]} \BibitemShut {NoStop}%
\bibitem [{\citenamefont {Takesaki}(2002)}]{TakesakiI}%
  \BibitemOpen
  \bibfield  {author} {\bibinfo {author} {\bibfnamefont {M.}~\bibnamefont
  {Takesaki}},\ }\href@noop {} {\emph {\bibinfo {title} {Theory of operator
  algebras. {I}}}},\ \bibinfo {series} {Encyclopaedia of Mathematical
  Sciences}, Vol.\ \bibinfo {volume} {124}\ (\bibinfo  {publisher}
  {Springer-Verlag},\ \bibinfo {address} {Berlin},\ \bibinfo {year} {2002})\
  pp.\ \bibinfo {pages} {xx+415}\BibitemShut {NoStop}%
\bibitem [{\citenamefont {{Bachmann}}\ \emph {et~al.}(2012)\citenamefont
  {{Bachmann}}, \citenamefont {{Michalakis}}, \citenamefont {{Nachtergaele}},\
  and\ \citenamefont {{Sims}}}]{BachmannMNS12}%
  \BibitemOpen
  \bibfield  {author} {\bibinfo {author} {\bibfnamefont {S.}~\bibnamefont
  {{Bachmann}}}, \bibinfo {author} {\bibfnamefont {S.}~\bibnamefont
  {{Michalakis}}}, \bibinfo {author} {\bibfnamefont {B.}~\bibnamefont
  {{Nachtergaele}}}, \ and\ \bibinfo {author} {\bibfnamefont {R.}~\bibnamefont
  {{Sims}}},\ }\bibfield  {title} {\enquote {\bibinfo {title} {{Automorphic
  Equivalence within Gapped Phases of Quantum Lattice Systems}},}\ }\href
  {\doibase 10.1007/s00220-011-1380-0} {\bibfield  {journal} {\bibinfo
  {journal} {Commun. Math. Phys.}\ }\textbf {\bibinfo {volume} {309}},\
  \bibinfo {pages} {835--871} (\bibinfo {year} {2012})},\ \Eprint
  {http://arxiv.org/abs/1102.0842} {arXiv:1102.0842 [math-ph]} \BibitemShut
  {NoStop}%
\bibitem [{\citenamefont {{Shi}}\ and\ \citenamefont {{Lu}}(2019)}]{Shi18}%
  \BibitemOpen
  \bibfield  {author} {\bibinfo {author} {\bibfnamefont {B.}~\bibnamefont
  {{Shi}}}\ and\ \bibinfo {author} {\bibfnamefont {Y.-M.}\ \bibnamefont
  {{Lu}}},\ }\bibfield  {title} {\enquote {\bibinfo {title} {{Characterizing
  topological orders by the information convex}},}\ }\href {\doibase
  10.1103/PhysRevB.99.035112} {\bibfield  {journal} {\bibinfo  {journal} {Phys.
  Rev. B}\ }\textbf {\bibinfo {volume} {99}},\ \bibinfo {pages} {035112}
  (\bibinfo {year} {2019})},\ \Eprint {http://arxiv.org/abs/1801.01519}
  {arXiv:1801.01519 [cond-mat.str-el]} \BibitemShut {NoStop}%
\bibitem [{\citenamefont {Almheiri}\ \emph {et~al.}(2015)\citenamefont
  {Almheiri}, \citenamefont {Dong},\ and\ \citenamefont
  {Harlow}}]{Almheiri2015}%
  \BibitemOpen
  \bibfield  {author} {\bibinfo {author} {\bibfnamefont {Ahmed}\ \bibnamefont
  {Almheiri}}, \bibinfo {author} {\bibfnamefont {Xi}~\bibnamefont {Dong}}, \
  and\ \bibinfo {author} {\bibfnamefont {Daniel}\ \bibnamefont {Harlow}},\
  }\bibfield  {title} {\enquote {\bibinfo {title} {Bulk locality and quantum
  error correction in {AdS/CFT}},}\ }\href {\doibase 10.1007/JHEP04(2015)163}
  {\bibfield  {journal} {\bibinfo  {journal} {Journal of High Energy Physics}\
  }\textbf {\bibinfo {volume} {2015}},\ \bibinfo {pages} {163} (\bibinfo {year}
  {2015})}\BibitemShut {NoStop}%
\bibitem [{\citenamefont {Datta}(2009)}]{4957651}%
  \BibitemOpen
  \bibfield  {author} {\bibinfo {author} {\bibfnamefont {N.}~\bibnamefont
  {Datta}},\ }\bibfield  {title} {\enquote {\bibinfo {title} {Min- and
  max-relative entropies and a new entanglement monotone},}\ }\href {\doibase
  10.1109/TIT.2009.2018325} {\bibfield  {journal} {\bibinfo  {journal} {IEEE
  Transactions on Information Theory}\ }\textbf {\bibinfo {volume} {55}},\
  \bibinfo {pages} {2816--2826} (\bibinfo {year} {2009})}\BibitemShut {NoStop}%
\bibitem [{\citenamefont {Jones}\ and\ \citenamefont
  {Sunder}(1997)}]{JonesSunder}%
  \BibitemOpen
  \bibfield  {author} {\bibinfo {author} {\bibfnamefont {V.}~\bibnamefont
  {Jones}}\ and\ \bibinfo {author} {\bibfnamefont {V.~S.}\ \bibnamefont
  {Sunder}},\ }\href {\doibase 10.1017/CBO9780511566219} {\emph {\bibinfo
  {title} {Introduction to subfactors}}},\ \bibinfo {series} {London
  Mathematical Society Lecture Note Series}, Vol.\ \bibinfo {volume} {234}\
  (\bibinfo  {publisher} {Cambridge University Press},\ \bibinfo {address}
  {Cambridge},\ \bibinfo {year} {1997})\ pp.\ \bibinfo {pages}
  {xii+162}\BibitemShut {NoStop}%
\bibitem [{\citenamefont {Shi}\ and\ \citenamefont
  {Lu}(2018)}]{PhysRevB.97.144106}%
  \BibitemOpen
  \bibfield  {author} {\bibinfo {author} {\bibfnamefont {Bowen}\ \bibnamefont
  {Shi}}\ and\ \bibinfo {author} {\bibfnamefont {Yuan-Ming}\ \bibnamefont
  {Lu}},\ }\bibfield  {title} {\enquote {\bibinfo {title} {Deciphering the
  nonlocal entanglement entropy of fracton topological orders},}\ }\href
  {\doibase 10.1103/PhysRevB.97.144106} {\bibfield  {journal} {\bibinfo
  {journal} {Phys. Rev. B}\ }\textbf {\bibinfo {volume} {97}},\ \bibinfo
  {pages} {144106} (\bibinfo {year} {2018})}\BibitemShut {NoStop}%
\bibitem [{Note1()}]{Note1}%
  \BibitemOpen
  \bibinfo {note} {More generally, there are additional constant terms which
  depends on the shape of the corners of the region.}\BibitemShut {Stop}%
\bibitem [{\citenamefont {Flammia}\ \emph {et~al.}(2009)\citenamefont
  {Flammia}, \citenamefont {Hamma}, \citenamefont {Hughes},\ and\ \citenamefont
  {Wen}}]{Flammia09}%
  \BibitemOpen
  \bibfield  {author} {\bibinfo {author} {\bibfnamefont {Steven~T.}\
  \bibnamefont {Flammia}}, \bibinfo {author} {\bibfnamefont {Alioscia}\
  \bibnamefont {Hamma}}, \bibinfo {author} {\bibfnamefont {Taylor~L.}\
  \bibnamefont {Hughes}}, \ and\ \bibinfo {author} {\bibfnamefont {Xiao-Gang}\
  \bibnamefont {Wen}},\ }\bibfield  {title} {\enquote {\bibinfo {title}
  {Topological entanglement {R}\'enyi entropy and reduced density matrix
  structure},}\ }\href {\doibase 10.1103/PhysRevLett.103.261601} {\bibfield
  {journal} {\bibinfo  {journal} {Phys. Rev. Lett.}\ }\textbf {\bibinfo
  {volume} {103}},\ \bibinfo {pages} {261601} (\bibinfo {year}
  {2009})}\BibitemShut {NoStop}%
\bibitem [{\citenamefont {Kato}\ \emph {et~al.}(2016)\citenamefont {Kato},
  \citenamefont {Furrer},\ and\ \citenamefont {Murao}}]{KatoFM16}%
  \BibitemOpen
  \bibfield  {author} {\bibinfo {author} {\bibfnamefont {Kohtaro}\ \bibnamefont
  {Kato}}, \bibinfo {author} {\bibfnamefont {Fabian}\ \bibnamefont {Furrer}}, \
  and\ \bibinfo {author} {\bibfnamefont {Mio}\ \bibnamefont {Murao}},\
  }\bibfield  {title} {\enquote {\bibinfo {title} {Information-theoretical
  analysis of topological entanglement entropy and multipartite
  correlations},}\ }\href {\doibase 10.1103/PhysRevA.93.022317} {\bibfield
  {journal} {\bibinfo  {journal} {Phys. Rev. A}\ }\textbf {\bibinfo {volume}
  {93}},\ \bibinfo {pages} {022317} (\bibinfo {year} {2016})},\ \Eprint
  {http://arxiv.org/abs/1505.01917} {arXiv:1505.01917 [quant-ph]} \BibitemShut
  {NoStop}%
\bibitem [{\citenamefont {Fiedler}\ \emph {et~al.}(2017)\citenamefont
  {Fiedler}, \citenamefont {Naaijkens},\ and\ \citenamefont
  {Osborne}}]{FiedlerNO17}%
  \BibitemOpen
  \bibfield  {author} {\bibinfo {author} {\bibfnamefont {Leander}\ \bibnamefont
  {Fiedler}}, \bibinfo {author} {\bibfnamefont {Pieter}\ \bibnamefont
  {Naaijkens}}, \ and\ \bibinfo {author} {\bibfnamefont {Tobias~J.}\
  \bibnamefont {Osborne}},\ }\bibfield  {title} {\enquote {\bibinfo {title}
  {Jones index, secret sharing and total quantum dimension},}\ }\href {\doibase
  10.1088/1367-2630/aa5c0c} {\bibfield  {journal} {\bibinfo  {journal} {New J.
  Phys.}\ }\textbf {\bibinfo {volume} {19}},\ \bibinfo {pages} {023039}
  (\bibinfo {year} {2017})},\ \Eprint {http://arxiv.org/abs/1608.02618}
  {arXiv:1608.02618 [quant-ph]} \BibitemShut {NoStop}%
\bibitem [{\citenamefont {Bowen~Shi}\ and\ \citenamefont {Kim}(2019)}]{BKK19}%
  \BibitemOpen
  \bibfield  {author} {\bibinfo {author} {\bibfnamefont {Kohtaro~Kato}\
  \bibnamefont {Bowen~Shi}}\ and\ \bibinfo {author} {\bibfnamefont {Isaac~H.}\
  \bibnamefont {Kim}},\ }\bibfield  {title} {\enquote {\bibinfo {title} {Fusion
  rules from entanglement},}\ }\href@noop {} {\bibfield  {journal} {\bibinfo
  {journal} {arXiv preprint}\ } (\bibinfo {year} {2019})}\BibitemShut {NoStop}%
\bibitem [{\citenamefont {Bravyi}(2008)}]{BravyiCEX}%
  \BibitemOpen
  \bibfield  {author} {\bibinfo {author} {\bibfnamefont {S.}~\bibnamefont
  {Bravyi}},\ }\href@noop {} {} (\bibinfo {year} {2008}),\ \bibinfo {note}
  {unpublished}\BibitemShut {NoStop}%
\bibitem [{\citenamefont {Kato}\ \emph {et~al.}(2014)\citenamefont {Kato},
  \citenamefont {Furrer},\ and\ \citenamefont {Murao}}]{KatoFM14}%
  \BibitemOpen
  \bibfield  {author} {\bibinfo {author} {\bibfnamefont {Kohtaro}\ \bibnamefont
  {Kato}}, \bibinfo {author} {\bibfnamefont {Fabian}\ \bibnamefont {Furrer}}, \
  and\ \bibinfo {author} {\bibfnamefont {Mio}\ \bibnamefont {Murao}},\
  }\bibfield  {title} {\enquote {\bibinfo {title} {Information-theoretical
  formulation of anyonic entanglement},}\ }\href {\doibase
  10.1103/PhysRevA.90.062325} {\bibfield  {journal} {\bibinfo  {journal} {Phys.
  Rev. A}\ }\textbf {\bibinfo {volume} {90}},\ \bibinfo {pages} {062325}
  (\bibinfo {year} {2014})},\ \Eprint {http://arxiv.org/abs/1310.4140}
  {arXiv:1310.4140 [quant-ph]} \BibitemShut {NoStop}%
\bibitem [{\citenamefont {{Zhou}}(2008)}]{PhysRevLett.101.180505}%
  \BibitemOpen
  \bibfield  {author} {\bibinfo {author} {\bibfnamefont {D.~L.}\ \bibnamefont
  {{Zhou}}},\ }\bibfield  {title} {\enquote {\bibinfo {title} {Irreducible
  multiparty correlations in quantum states without maximal rank},}\ }\href
  {\doibase 10.1103/PhysRevLett.101.180505} {\bibfield  {journal} {\bibinfo
  {journal} {Phys. Rev. Lett.}\ }\textbf {\bibinfo {volume} {101}},\ \bibinfo
  {pages} {180505} (\bibinfo {year} {2008})}\BibitemShut {NoStop}%
\bibitem [{\citenamefont {{Weis}}(2014)}]{Weis10}%
  \BibitemOpen
  \bibfield  {author} {\bibinfo {author} {\bibfnamefont {S.}~\bibnamefont
  {{Weis}}},\ }\bibfield  {title} {\enquote {\bibinfo {title} {{Information
  topologies on non-commutative state spaces}},}\ }\href@noop {} {\bibfield
  {journal} {\bibinfo  {journal} {J. Conv. Anal.}\ }\textbf {\bibinfo {volume}
  {21}},\ \bibinfo {pages} {339--399} (\bibinfo {year} {2014})}\BibitemShut
  {NoStop}%
\bibitem [{\citenamefont {Scholz}\ and\ \citenamefont
  {Werner}(2008)}]{ScholzWerner08}%
  \BibitemOpen
  \bibfield  {author} {\bibinfo {author} {\bibfnamefont {V.~B.}\ \bibnamefont
  {Scholz}}\ and\ \bibinfo {author} {\bibfnamefont {R.~F.}\ \bibnamefont
  {Werner}},\ }\href@noop {} {\enquote {\bibinfo {title} {Tsirelson's
  problem},}\ } (\bibinfo {year} {2008}),\ \bibinfo {note} {preprint},\ \Eprint
  {http://arxiv.org/abs/0812.4305} {arXiv:0812.4305} \BibitemShut {NoStop}%
\bibitem [{\citenamefont {Naaijkens}(2017)}]{Naaijkens17}%
  \BibitemOpen
  \bibfield  {author} {\bibinfo {author} {\bibfnamefont {Pieter}\ \bibnamefont
  {Naaijkens}},\ }\href {\doibase 10.1007/978-3-319-51458-1} {\emph {\bibinfo
  {title} {Quantum spin systems on infinite lattices: a concise
  introduction}}},\ \bibinfo {series} {Lecture Notes in Physics}, Vol.\
  \bibinfo {volume} {933}\ (\bibinfo  {publisher} {Springer, Cham},\ \bibinfo
  {year} {2017})\ pp.\ \bibinfo {pages} {xi+177},\ \Eprint
  {http://arxiv.org/abs/1311.2717} {arXiv:1311.2717 [math-ph]} \BibitemShut
  {NoStop}%
\bibitem [{\citenamefont {Drinfel'd}(1987)}]{Drinfeld87}%
  \BibitemOpen
  \bibfield  {author} {\bibinfo {author} {\bibfnamefont {V.~G.}\ \bibnamefont
  {Drinfel'd}},\ }\bibfield  {title} {\enquote {\bibinfo {title} {Quantum
  groups},}\ }in\ \href@noop {} {\emph {\bibinfo {booktitle} {Proceedings of
  the {I}nternational {C}ongress of {M}athematicians, {V}ol. 1, 2 ({B}erkeley,
  {C}alif., 1986)}}}\ (\bibinfo  {publisher} {Amer. Math. Soc.},\ \bibinfo
  {address} {Providence, RI},\ \bibinfo {year} {1987})\ pp.\ \bibinfo {pages}
  {798--820}\BibitemShut {NoStop}%
\bibitem [{\citenamefont {Dijkgraaf}\ \emph {et~al.}(1991)\citenamefont
  {Dijkgraaf}, \citenamefont {Pasquier},\ and\ \citenamefont
  {Roche}}]{DijkgraafPR91}%
  \BibitemOpen
  \bibfield  {author} {\bibinfo {author} {\bibfnamefont {R.}~\bibnamefont
  {Dijkgraaf}}, \bibinfo {author} {\bibfnamefont {V.}~\bibnamefont {Pasquier}},
  \ and\ \bibinfo {author} {\bibfnamefont {P.}~\bibnamefont {Roche}},\
  }\bibfield  {title} {\enquote {\bibinfo {title} {Quasi {H}opf algebras, group
  cohomology and orbifold models},}\ }\href {\doibase
  10.1016/0920-5632(91)90123-V} {\bibfield  {journal} {\bibinfo  {journal}
  {Nuclear Phys. B Proc. Suppl.}\ }\textbf {\bibinfo {volume} {18B}},\ \bibinfo
  {pages} {60--72} (\bibinfo {year} {1991})},\ \bibinfo {note} {recent advances
  in field theory (Annecy-le-Vieux, 1990)}\BibitemShut {NoStop}%
\bibitem [{\citenamefont {Bombin}\ and\ \citenamefont
  {Martin-Delgado}(2008)}]{BombinMD08}%
  \BibitemOpen
  \bibfield  {author} {\bibinfo {author} {\bibfnamefont {H.}~\bibnamefont
  {Bombin}}\ and\ \bibinfo {author} {\bibfnamefont {M.~A.}\ \bibnamefont
  {Martin-Delgado}},\ }\bibfield  {title} {\enquote {\bibinfo {title} {Family
  of non-{A}belian {K}itaev models on a lattice: {T}opological condensation and
  confinement},}\ }\href {\doibase 10.1103/PhysRevB.78.115421} {\bibfield
  {journal} {\bibinfo  {journal} {Phys. Rev. B}\ }\textbf {\bibinfo {volume}
  {78}},\ \bibinfo {pages} {115421} (\bibinfo {year} {2008})},\ \Eprint
  {http://arxiv.org/abs/0712.0190} {arXiv:0712.0190 [cond-mat.str-el]}
  \BibitemShut {NoStop}%
\bibitem [{\citenamefont {Shi}(2018)}]{BShipriv}%
  \BibitemOpen
  \bibfield  {author} {\bibinfo {author} {\bibfnamefont {B.}~\bibnamefont
  {Shi}},\ }\href@noop {} {\bibfield  {journal} {\bibinfo  {journal} {private
  communication}\ } (\bibinfo {year} {2018})}\BibitemShut {NoStop}%
\bibitem [{\citenamefont {Preskill}(2004)}]{Preskill}%
  \BibitemOpen
  \bibfield  {author} {\bibinfo {author} {\bibfnamefont {John}\ \bibnamefont
  {Preskill}},\ }\href@noop {} {\enquote {\bibinfo {title} {Topological quantum
  computation},}\ } (\bibinfo {year} {2004}),\ \bibinfo {note} {chapter 9 of
  lecture notes available at
  \url{http://www.theory.caltech.edu/people/preskill/ph229/}}\BibitemShut
  {NoStop}%
\bibitem [{\citenamefont {Cui}\ \emph {et~al.}(2015)\citenamefont {Cui},
  \citenamefont {Hong},\ and\ \citenamefont {Wang}}]{CuiHW2015}%
  \BibitemOpen
  \bibfield  {author} {\bibinfo {author} {\bibfnamefont {Shawn~X.}\
  \bibnamefont {Cui}}, \bibinfo {author} {\bibfnamefont {Seung-Moon}\
  \bibnamefont {Hong}}, \ and\ \bibinfo {author} {\bibfnamefont {Zhenghan}\
  \bibnamefont {Wang}},\ }\bibfield  {title} {\enquote {\bibinfo {title}
  {Universal quantum computation with weakly integral anyons},}\ }\href
  {\doibase 10.1007/s11128-015-1016-y} {\bibfield  {journal} {\bibinfo
  {journal} {Quantum Information Processing}\ }\textbf {\bibinfo {volume}
  {14}},\ \bibinfo {pages} {2687--2727} (\bibinfo {year} {2015})}\BibitemShut
  {NoStop}%
\bibitem [{Note2()}]{Note2}%
  \BibitemOpen
  \bibinfo {note} {Equivalently one could look at the space $\protect
  \operatorname {Hom}(\iota , \tau ^{\otimes n})$, which could be interpreted
  as all the ways to create $\tau $ anyons out of the vacuum.}\BibitemShut
  {Stop}%
\bibitem [{\citenamefont {Naaijkens}(2018)}]{Naaijkens18}%
  \BibitemOpen
  \bibfield  {author} {\bibinfo {author} {\bibfnamefont {Pieter}\ \bibnamefont
  {Naaijkens}},\ }\bibfield  {title} {\enquote {\bibinfo {title} {Subfactors
  and quantum information theory},}\ }in\ \href {\doibase
  10.1090/conm/717/14453} {\emph {\bibinfo {booktitle} {Mathematical problems
  in quantum physics}}},\ \bibinfo {series} {Contemp. Math.}, Vol.\ \bibinfo
  {volume} {717}\ (\bibinfo  {publisher} {Amer. Math. Soc., Providence, RI},\
  \bibinfo {year} {2018})\ pp.\ \bibinfo {pages} {257--279},\ \Eprint
  {http://arxiv.org/abs/1704.05562} {arXiv:1704.05562 [math-ph]} \BibitemShut
  {NoStop}%
\bibitem [{\citenamefont {Cha}\ \emph {et~al.}(2019)\citenamefont {Cha},
  \citenamefont {Naaijkens},\ and\ \citenamefont {Nachtergaele}}]{ChaNN18}%
  \BibitemOpen
  \bibfield  {author} {\bibinfo {author} {\bibfnamefont {Matthew}\ \bibnamefont
  {Cha}}, \bibinfo {author} {\bibfnamefont {Pieter}\ \bibnamefont {Naaijkens}},
  \ and\ \bibinfo {author} {\bibfnamefont {Bruno}\ \bibnamefont
  {Nachtergaele}},\ }\bibfield  {title} {\enquote {\bibinfo {title} {On the
  stability of charges in infinite quantum spin systems},}\ }\href@noop {}
  {\bibfield  {journal} {\bibinfo  {journal} {Commun. Math. Phys.}\ } (\bibinfo
  {year} {2019})},\ \bibinfo {note} {in press. Preprint arXiv:1804.03203},\
  \Eprint {http://arxiv.org/abs/1804.03203} {arXiv:1804.03203 [math-ph]}
  \BibitemShut {NoStop}%
\end{thebibliography}
\end{document}